\documentclass[11pt]{article}
\RequirePackage{amssymb,epsf,epic,amsmath}

\def\denseformat{
\setlength{\textheight}{9in}
\setlength{\textwidth}{6.9in}
\setlength{\evensidemargin}{-0.2in}
\setlength{\oddsidemargin}{-0.2in}
\setlength{\headsep}{10pt}
\setlength{\topmargin}{-0.3in}
\setlength{\columnsep}{0.375in}
\setlength{\itemsep}{0pt}
}

\newtheorem{theorem}{Theorem}[section]

\newtheorem{claim}[theorem]{Claim}
\newtheorem{lemma}[theorem]{Lemma}

\newtheorem{corollary}[theorem]{Corollary}

\newtheorem{comment}[theorem]{Comment}

\newtheorem{problem}[theorem]{Problem}
\newtheorem{observation}[theorem]{Observation}

\def\boldhead#1:{\par\vskip 7pt\noindent{\bf #1:}\hskip 10pt}
\def\ithead#1:{\par\vskip 7pt\noindent{\it #1:}\hskip 10pt}

\def\inline#1:{\par\vskip 7pt\noindent{\bf #1:}\hskip 10pt}
\def\midinline#1:{\par\noindent{\bf #1:}\hskip 10pt}
\def\dnsinline#1:{\par\vskip -7pt\noindent{\bf #1:}\hskip 10pt}
\def\ddnsinline#1:{\newline{\bf #1:}\hskip 10pt}
\def\largeinline#1:{\par\vskip 7pt\noindent{\large\bf #1:}\hskip 10pt}
\long\def\comment #1\commentend{}
\long\def\commhide #1\commhideend{}
\long\def\commfull #1\commend{#1}
\long\def\commabs #1\commenda{}
\long\def\commtim #1\commendt{#1}
\long\def\commb #1\commbend{}
\long\def\commedit #1\commeditend{} 

\long\def\commB #1\commBend{}       

\long\def\commex #1\commexend{}     

\long\def\commsiena #1\commsienaend{}  
                                         
\long\def\commBI #1\commBIend{}  
                                         
\long\def\CProof #1\CQED{}

\def\blackslug{\hbox{\hskip 1pt \vrule width 4pt height 8pt
    depth 1.5pt \hskip 1pt}}
\def\QED{\quad\blackslug\lower 8.5pt\null\par}
\def\inQED{\quad\quad\blackslug}

\long\def\PPP#1{\noindent{\bf Proof:}{ #1}{\quad\blackslug\lower 8.5pt\null}}

\long\def\denspar #1\densend
{#1}

\newif\ifnotesw\noteswtrue
   {\ifnotesw\marginpar[\hfill\(\top\)]{\(\top\)}\fi}%
   {\ifnotesw\marginpar[\hfill\(\bot\)]{\(\bot\)}\fi}

\newcommand{\mnote}[1]%
    {\ifnotesw\marginpar%
        [{\scriptsize\it\begin{minipage}[t]{\marginparwidth}
        \raggedleft#1%
                        \end{minipage}}]%
        {\scriptsize\it\begin{minipage}[t]{\marginparwidth}
        \raggedright#1%
                        \end{minipage}}%
    \fi}

\def\tR{{\tilde R}}

\def\MathF{\hbox{\rm I\kern-2pt F}}
\def\MathP{\hbox{\rm I\kern-2pt P}}
\def\MathR{\hbox{\rm I\kern-2pt R}}
\def\MathZ{\hbox{\sf Z\kern-4pt Z}}
\def\MathN{\hbox{\rm I\kern-2pt I\kern-3.1pt N}}
\def\MathC{\hbox{\rm \kern0.7pt\raise0.8pt\hbox{\footnotesize I}
\kern-4.2pt C}}
\def\MathQ{\hbox{\rm I\kern-6pt Q}}

\newsavebox{\ttop}\newsavebox{\bbot}

\def\eps{\epsilon}

\def\nin{{~\not \in~}}

\denseformat

\begin{document}

\def\varrhol{{et al.~}}
\def\Sp{\mathit{Sp}}
\def\ReadDg{\mathit{Read\_Edge}}
\def\wmax{{{\hat \omega}}}
\def\ttl{\mathit{ttl}}
\def\Rnd{\mathit{Round}}
\def\deg{\mathit{deg}}
\def\ctr{\mathit{ctr}}
\def\ScDgs{\mathit{Scanned\_Edges}}
\def\SCD{\mathit{SCANNED}}
\def\NOTSCD{\mathit{NOTSCANNED}}
\def\CRASH{\mathit{CRASH}}
\def\CR{\mathit{CRASH}}
\def\CRASHED{\mathit{CRASHED}}
\def\SyncIncr{\mathit{Sync\_Incr}}
\def\actdeg{\mathit{actdeg}}
\def\OLD{\mathit{OLD}}
\def\NEW{\mathit{NEW}}
\def\mark{\mathit{mark}}
\def\AsyncRnd{\mathit{Async\_Rnd}}
\def\status{\mathit{status}}
\def\scstat{\mathit{scan\_status}}
\def\lab{\mathit{label}}
\def\seclabel{\mathit{sec\_label}}
\def\own{\mathit{own}}
\def\SELF{\mathit{SELF}}
\def\PEER{\mathit{PEER}}
\def\DynRnd{\mathit{Dyn\_Rnd}}
\def\Delete{\mathit{Delete}}
\def\XReplace{\mathit{XReplace}}
\def\UpdLab{\mathit{Update\_Label}}
\def\Crash{\mathit{Crash}}
\def\CrashLoop{\mathit{CrashLoop}}
\def\CrashItern{\mathit{CrashItern}}
\def\fetched{\mathit{fetched}}
\def\TRUE{\mathit{TRUE}}
\def\FALSE{\mathit{FALSE}}
\def\crashstat{\mathit{crash\_status}}
\def\wmax{\hat{\omega}}
\def\END{\mathit{END}}
\def\te{\tilde{e}}
\def\tu{\tilde{u}}
\def\che{\check{e}}
\def\chu{\check{u}}
\def\Old{\mathit{Old}}
\def\tR{\tilde{R}}
\def\eps{{\epsilon}}
\def\UpLab{\mathit{Update\_Label}}
\def\chP{\check{P}}
\def\VOID{\mathit{VOID}}

\newcommand {\ignore} [1] {}

\title{An Optimal-Time Construction of Euclidean \\Sparse Spanners with Tiny Diameter}

\author{
Shay Solomon\thanks{
        Department of Computer Science,
        Ben-Gurion University of the Negev, POB 653, Beer-Sheva 84105, Israel.
         E-mail: {\tt \{shayso\}@cs.bgu.ac.il}
        \newline This research has been supported by the
                 Clore Fellowship grant No.\ 81265410.
        \newline Partially supported by the
       Lynn and William Frankel
        Center for Computer Sciences.}}
\date{\empty}

\date{\empty}

\begin{titlepage}
\def\thepage{}
\maketitle

\begin{abstract}
In STOC'95 \cite{ADMSS95} Arya et al.\ showed that for any set of $n$ points in $\mathbb R^d$, 
a $(1+\epsilon)$-spanner with diameter at most 2 (respectively, 3) and $O(n \log n)$ edges (resp., $O(n \log \log n)$ edges)
can be built in $O(n \log n)$ time. Moreover, it was shown in \cite{ADMSS95,NS07} that for any $k \ge 4$, one can
build in $O(n (\log n) 2^k \alpha_k(n))$ time
a $(1+\epsilon)$-spanner with diameter at most $2k$ and $O(n 2^k \alpha_k(n))$ edges. 
The function $\alpha_k$ is the inverse of a certain function at the $\lfloor k/2 \rfloor$th level of the primitive recursive hierarchy,
where $\alpha_0(n) = \lceil n/2 \rceil, \alpha_1(n) = \left\lceil \sqrt{n} \right\rceil,
\alpha_2(n) = \lceil \log{n} \rceil,
\alpha_3(n) = \lceil \log\log{n} \rceil,
\alpha_4(n) = \log^* n$, 
 \ldots, etc.
It is also known \cite{NS07} that if one allows quadratic time then these bounds can be improved. Specifically, for any $k \ge 4$, a $(1+\epsilon)$-spanner with diameter at most $k$ and 
$O(n k \alpha_k(n))$ edges can be constructed in $O(n^2)$ time \cite{NS07}.

A major open problem in this area is whether one can construct within time $O(n \log n + n k \alpha_k(n))$  
a $(1+\epsilon)$-spanner  with diameter at most $k$
and $O(n k \alpha_k(n))$ edges.
 In this paper we answer this question in the affirmative.
Moreover, in fact, we provide a stronger result. Specifically, we show that for any $k \ge 4$, a $(1+\epsilon)$-spanner with 
diameter at most $k$ and $O(n \alpha_k(n))$ edges can be built in optimal time $O(n \log n)$.
The tradeoff between the diameter and number of edges of our
spanners is tight up to constant factors in the entire range of parameters. 
\end{abstract}
\end{titlepage}

\pagenumbering {arabic} 

\section{Introduction}
\subsection{Euclidean Spanners} \label{s:first}
Consider a set $S$ of $n$ points in $\mathbb R^d$ and a number $t \ge 1$.
A graph $G = (S,E)$ in which the weight $w(x,y)$ of each edge $e =(x,y)\in E$ 
is equal to the Euclidean distance $\|x-y\|$ between $x$ and $y$ is called a \emph{geometric graph}.
We say that the graph $G$ is a \emph{$t$-spanner} for $S$ if for every pair $p,q \in S$ of distinct points,
there exists a path in $G$ between $p$ and $q$ whose weight\footnote{The \emph{weight} of a path is defined to be the sum of all edge weights  in it.} is at most $t$ times the
Euclidean distance $\|p-q\|$ between $p$ and $q$. 
Such a path 
is called a \emph{$t$-spanner path}. 
The problem of constructing
Euclidean spanners 
has been studied intensively
over the past two decades
\cite{Chew86,KG92,ADDJS93,CK93,DN94,ADMSS95,DNS95,AS97,RS98,AWY05,CG06,DES08}. 
(See also the recent book by Narasimhan and Smid \cite{NS07},
and the references therein.) 
Euclidean spanners find applications in geometric approximation algorithms, network topology design, geometric distance oracles, distributed systems, design of parallel machines, and other areas \cite{DN94,LNS98,RS98,GLNS02,GNS05,GLNS08,HP00,MP00}.

Spanners are important geometric structures, since they enable approximation of the complete Euclidean
graph in a much more economical form. First and foremost, a spanner should be \emph{sparse}, meaning that
it can have only a small (ideally, linear) number of edges. However, at the same time, the spanner is required to
preserve some fundamental properties of the underlying complete graph.
In particular, for some practical applications (e.g., in network routing protocols) it is desirable that the spanner achieves a small \emph{diameter}, that is, 
for every pair $p,q \in S$ of distinct points there should be a $t$-spanner path that consists of a small
number of edges \cite{AMS94,AM04,AWY05,CG06,DES08}. 

  In a seminal STOC'95 paper \cite{ADMSS95}, Arya et al.\ showed that for any set of $n$ points in $\mathbb R^d$ one can build
 in $O(n \log n)$ time a
 $(1+\epsilon)$-spanner with diameter
at most 2 and $O(n \log n)$ edges, and another one with diameter at most 3 and $O(n \log \log n)$  
 edges.
 Moreover, it was shown in \cite{ADMSS95,NS07} that for any $k \ge 4$,
one can build in $O(n (\log n) 2^k \alpha_k(n))$ time a $(1+\epsilon)$-spanner with diameter at most $2k$ and $O(n 2^k \alpha_k(n))$ edges.
The function $\alpha_k$ is the inverse of a certain function at the $\lfloor k/2 \rfloor$th level of the primitive recursive hierarchy,
where $\alpha_0(n) = \lceil n/2 \rceil,
\alpha_1(n) = \left\lceil \sqrt{n} \right\rceil,
\alpha_2(n) = \lceil \log{n} \rceil,
\alpha_3(n) = \lceil \log\log{n} \rceil,
\alpha_4(n) = \log^* n, \alpha_5(n) = \lfloor \frac{1}{2} \log^*n \rfloor$,
 $\ldots$, etc. Roughly speaking, for $k \ge 2$ the function $\alpha_{k}$ is close to $\log$ with $\lfloor \frac{k-2}{2} \rfloor$ stars.
(See Section \ref{s:value} for the formal definition of this function.)
It is also known \cite{NS07} that if one allows quadratic time then these bounds can be improved. Specifically, for any $k \ge 4$, a $(1+\epsilon)$-spanner with diameter at most $k$ and 
$O(n k \alpha_k(n))$ edges can be constructed in $O(n^2)$ time \cite{NS07}.

If one wishes to produce spanners with $o(n \log \log n)$ edges
but is willing to spend only $O(n \log n)$ time, 
then none of the above constructions of \cite{ADMSS95,NS07}
is of any help. However, there is another construction of spanners
that can be used \cite{CG06}. 
Specifically, Chan and Gupta \cite{CG06} showed that 
there exists a $(1+\eps)$-spanner 
with diameter $O(\alpha(m,n))$
and $O(m)$ edges\footnote{The construction of \cite{CG06}
applies, in fact, to doubling metrics. This tradeoff translates into a tradeoff of $O(k)$ versus $O(n \alpha_k(n))$ between the diameter and number of edges.}, where $\alpha(\cdot,\cdot)$ is the 
inverse-Ackermann function.
Moreover, the construction of \cite{CG06}
can be implemented in $O(n \log n)$ time.
The drawback of this construction, though, is that the constant factor hidden within the $O$-notation of the diameter bound is large.
In particular, it does not provide a spanner whose diameter is, say, smaller than 50.

A major open problem in this area\footnote{It appears as open problem number 19 in the list of
open problems in the treatise of Narasimhan and Smid \cite{NS07} on Euclidean spanners; see p.\ 481.} is whether one can construct in time $O(n \log n + n k \alpha_k(n))$ 
a $(1+\epsilon)$-spanner  with diameter at most $k$
and $O(n k \alpha_k(n))$ edges,
 for any $k \ge 4$. In this paper we answer this question in the affirmative.
In fact, we provide a stronger result. Specifically, we show that a $(1+\epsilon)$-spanner with 
diameter at most $k$ and $O(n \alpha_k(n))$ edges can be built in $O(n \log n)$ time.

Note that our tradeoff improves all previous results in a number of senses. In comparison to the construction of \cite{NS07}
that requires a quadratic running time, our construction is (1) drastically faster, and (2) produces a spanner that is sparser
by a factor of $k$. In comparison to the result of \cite{ADMSS95,NS07} that for any $k \ge 4$ produces 
in time $O(n (\log n) 2^k \alpha_k(n))$, a $(1+\epsilon)$-spanner
with diameter at most $2k$ and $O(n 2^k \alpha_k(n))$ edges, our construction
has (1) a diameter half as large,
(2) is faster
by a factor of $2^k \alpha_k(n)$, and (3) produces a spanner that is sparser by
a factor of $2^k$. Finally, in comparison to the construction of \cite{CG06}, the diameter of our construction
is smaller by a significant constant factor.
(See Table \ref{tab1} for a concise comparison of previous and new results.)

\begin{table}
\begin{center}
\begin{tabular}{|c|c|c|c|c|}
\hline  & \cite{NS07} &  \cite{ADMSS95,NS07}&  \cite{CG06} & {\bf New} \\
\hline  Diameter & $k$ & $2k$ & $O(k)$ &
{\boldmath $k$} \\
\hline Number of edges & $O(n k \alpha_k(n))$   & $O(n 2^k \alpha_k(n))$ &  $O(n \alpha_k(n))$ & {\boldmath $O(n \alpha_k(n))$} \\
\hline Running time & $O(n^2)$ & $O(n (\log n) 2^k \alpha_k(n))$ & $O(n \log n)$    & {\boldmath $O(n \log n)$} \\
\hline
\end{tabular}
\end{center}
\caption[]{ \label{tab1} \footnotesize A concise comparison of previous and new
results on Euclidean spanners. Our results are indicated by bold fonts.
It is assumed that $n$ and $k$ are arbitrary parameters, with $k \ge 4$.
} 
\end{table}
 
There are two particular values of $k$ that deserve special attention.
First, for $k=4$ our tradeoff shows that one can build in $O(n \log n)$ time a $(1+\epsilon)$-spanner
with diameter at most 4 and $O(n \log^* n)$ edges. 
This result provides the \emph{first subquadratic-time construction} of $(1+\epsilon)$-spanners with diameter at most $7$ and $o(n \log \log n)$ edges. 
Second, for $k = 2\alpha(n)+2$, our tradeoff gives rise to a diameter at most $2\alpha(n)+2$ and $O(n \alpha_{2\alpha(n)+2}(n)) = O(n)$ edges. In all the previous works of \cite{ADMSS95,CG06,NS07}, a construction of $(1+\eps)$-spanners
with diameter $O(\alpha(n))$ and $O(n)$ edges was also provided. However,
the constants hidden within the $O$-notation of the diameter bound  in the corresponding constructions of \cite{ADMSS95,CG06,NS07} are significantly larger than $2$.
Since $\alpha(n) \le 4$ for all practical
applications, this improvement on the diameter bound is of practical importance.  

Our tradeoff is \emph{tight in all respects}.
Indeed, the upper bound of $O(n \log n)$ on the running time of our construction holds in the \emph{algebraic computation-tree model}.
A matching lower bound is given in \cite{CDS01}.
In addition, Chan and Gupta \cite{CG06} 
proved that there exists a set of $n$ points
on the $x$-axis for which any $(1+\epsilon)$-spanner with at most $m$ edges must have a diameter at least $\alpha(m,n) -4$.
This lower bound (cf.\ Corollary 4.1 therein \cite{CG06}) 
implies that our tradeoff of $k$ versus $O(n \alpha_k(n))$ between the diameter and number of edges
cannot be improved\footnote{See also open problem 20 on p.\ 481 in \cite{NS07}, and  the corresponding solution \cite{CG06,Smid10}.}
 by more than constant factors even for 1-dimensional spaces.
 

\ignore{
Finally, note that the lower bound of \cite{CG06} does not preclude the existence
of Euclidean Steiner spanners\footnote{A \emph{Euclidean Steiner spanner} for a point set $S$ is a spanner that may contain additional \emph{Steiner points}, i.e.,
points that do not belong to the original point set $S$.} with diameter $o(\alpha(m,n))$ and $o(m)$ 
edges (or equivalently, with diameter $o(k)$ and $o(n \alpha_k(n))$ edges).
We extend (see Appendix \ref{s:lower}) the lower bound of \cite{CG06} to Euclidean Steiner spanners and show that as far as the diameter and number of edges are concerned, 
\emph{Steiner points do not help}. Consequently, our construction of Euclidean spanners 
cannot be improved
even if one allows the spanner to employ (arbitrarily many) Steiner points. 
}
   
\subsection{1-Spanners for Tree Metrics}
The tree metric induced by an arbitrary $n$-vertex weighted tree $T$ is denoted by $M_T$. 
A spanning subgraph $G$ of $M_T$ is said to be a \emph{1-spanner} for $T$, if for every pair of vertices, their distance in $G$ is equal to their distance in $T$.
Let $P_n$ be the unweighted path graph on $n$ vertices.

In a classical STOC'82 paper \cite{Yao82}, Yao showed that there exists a 1-spanner\footnote{Yao stated this problem in terms of
partial sums.} 
for $P_n$ with 
diameter $O(\alpha(m,n))$ and $O(m)$ edges, for any $m \ge n$. 
Chazelle \cite{Chaz87} extended the result of \cite{Yao82} to arbitrary tree metrics, and presented an $O(m)$-time algorithm for computing a
1-spanner realizing this tradeoff. Thorup \cite{Thor97} provided an alternative algorithmic proof of Chazelle's result \cite{Chaz87}, and,
in addition, 
devised an efficient parallel algorithm for computing such a 1-spanner. 
In all these constructions \cite{Yao82,Chaz87,Thor97}, the constants hidden within the $O$-notation
of the diameter bound are significantly larger than 2. In particular,
none of these constructions can produce a spanner whose diameter is, say, smaller than 25.
  
Alon and Schieber \cite{AS87} and Bodlaender et al.\ \cite{BTS94} 
 independently showed that a 1-spanner for $P_n$ with diameter at most $k$ and $O(n \alpha_k(n))$ edges can be built in $O(n \alpha_k(n))$ time, for any \emph{constant}  $k \ge 2$. 
The constructions of \cite{AS87} and \cite{BTS94} were also extended to arbitrary tree metrics.
Specifically, Alon and Schieber \cite{AS87} showed that for any tree metric, a 1-spanner with diameter at most $2k$ 
(rather than $k$) and $O(n \alpha_k(n))$ edges can be built in $O(n \alpha_k(n))$ time.
Also, they managed to reduce the diameter bound from $2k$ back to $k$ in the particular cases of $k=2$ and $k=3$.
Bodlaender et al.\ \cite{BTS94} devised a construction of 1-spanners for arbitrary tree metrics
with diameter at most $k$ and $O(n \alpha_k(n))$ edges.
However, the question of whether this construction of Bodlaender et al.\ can be implemented efficiently was left open in \cite{BTS94},
and remained open prior to this work. 

Narasimhan and Smid \cite{NS07} extended the constructions 
of \cite{AS87} and \cite{BTS94}
to super-constant values of $k$. 
Specifically, they showed that for any $k \ge 4$ and any tree metric, a 1-spanner 
with diameter at most $2k$ (respectively, $k$) and $O(n 2^k \alpha_k(n))$ (resp., $O(n k \alpha_k(n))$ 
edges can be built in
time $O(n (\log n) 2^k \alpha_k(n))$ (resp., $O(n^2)$). 

On the way to our results for Euclidean spanners, we have improved
the constructions of \cite{AS87,BTS94,NS07} and devised an
$O(n\alpha_k(n))$-time algorithm that builds 1-spanners for arbitrary tree metrics
with diameter at most $k$ and $O(n \alpha_k(n))$ edges, for \emph{any} $k \ge 4$.
(See Table \ref{tab2} for a comparison of previous and new results.)
\begin{table}
\begin{center}
\begin{tabular}{|c|c|c|c|}
\hline  & Diameter & Number of edges &  Running time \\
\hline  \cite{Chaz87,Thor97} & $O(k)$, for any $k$ & $O(n \alpha_k(n))$ & $O(n \alpha_k(n))$ \\
\hline \cite{AS87}& $2k$, for a \emph{constant} $k \ge 4$ & $O(n \alpha_k(n))$ & $O(n \alpha_k(n))$ \\ 
\hline \cite{NS07} & $2k$, for any $k \ge 4$ & $O(n 2^k \alpha_k(n))$ & $O(n (\log n) 2^k \alpha_k(n))$  \\
\hline \cite{BTS94,NS07} & $k$, for any $k \ge 4$ & $O(n k \alpha_k(n))$ & $O(n^2)$  \\
\hline {\bf New} & {\boldmath $k$}{\bf , for any} {\boldmath $k \ge 4$} & {\boldmath $O(n \alpha_k(n))$} & {\boldmath $O(n \alpha_k(n))$}  \\
\hline
\hline \cite{AS87} & 7 &  $O(n \log \log n)$ & $O(n \log \log n)$  \\
\hline \cite{BTS94,NS07} & 4 & $O(n \log^* n)$ & $O(n^2)$  \\
\hline {\bf New} & {\boldmath $4$} & {\boldmath $O(n \log^* n)$} & {\boldmath $O(n \log^* n)$}  \\
\hline
\hline \cite{Chaz87,Thor97,AS87,NS07} & $O(\alpha(n))$ & $O(n)$ & $O(n)$  \\
\hline {\bf New} & {\boldmath $2\alpha(n)+2$} & {\boldmath$O(n)$} & {\boldmath $O(n)$}  \\
\hline
\end{tabular}
\end{center}
\caption[]{ \label{tab2} \footnotesize A comparison of previous and new
results on 1-spanners for arbitrary tree metrics. Our results are indicated by bold fonts.
\ignore{The next
three rows indicate the resulting degree ($\Delta$), diameter
$(\Lambda)$ and lightness $(\Psi)$. The number of edges used in all
constructions is $O(n)$. To save space, the $O$-notation is omitted
everywhere except for the exponents. The letters $\delta$ and
$\zeta$ stand for arbitrarily small positive constants.}
} 
\end{table}
The running time of our algorithm is \emph{linear} in the number of edges of the resulting spanners.
Also, it was proved in \cite{AS87} that any 1-spanner for $P_n$ with diameter at most $k$ must have at least $\Omega(n \alpha_k(n))$ edges. This lower bound implies that the tradeoff 
between the diameter and number of edges of our spanners is tight in the entire range of parameters.

The problem of constructing 1-spanners for tree metrics is a natural one, and, not surprisingly, has also been studied in more
general settings, such as planar
metrics \cite{Thor95}, general metrics \cite{Thor92}, and general graphs \cite{BGJRW09}. (See
also Chapter 12 in \cite{NS07} for an excellent survey on this
problem.)
This problem is also closely related to the extremely well-studied problem of
computing partial-sums. (See the papers of Tarjan \cite{Tarj79}, Yao \cite{Yao82},
Chazelle and Rosenberg \cite{CR91}, ~P\u{a}tra\c{s}cu and Demaine
\cite{PD04}, and the references therein.) For a discussion about the
relationship between these two problems see the introduction
of \cite{AWY05}.
 We demonstrate that our construction of 1-spanners for tree metrics is useful for improving key results 
in the context of Euclidean spanners. We anticipate that this construction 
would be found useful in the context of partial sums problems, and
for other
applications such as those discussed in \cite{BTS94,BGJRW09}. 
Finally, we believe that regardless of its
applications, this construction is of
independent interest.

\subsection{Lower Bounds for Euclidean Steiner Spanners}
The lower bound of Chan and Gupta \cite{CG06} on the tradeoff between the diameter and number of edges
of Euclidean spanners was mentioned in Section \ref{s:first}.
Formally, it states that for any $\eps > 0$,
there exists a set of $n$ points on the $x$-axis, where $n$ is an arbitrary power of two,
for which any $(1+\eps)$-spanner
with at most $m$ edges has diameter at least $\alpha(m,n) -4$.
Consequently, the corresponding upper bound construction of \cite{CG06} is optimal up to constant factors.
However, the lower bound of Chan and Gupta \cite{CG06} does not preclude the existence
of Euclidean Steiner spanners\footnote{A \emph{Euclidean Steiner spanner} for a point set $S$ is a spanner that may contain additional \emph{Steiner points}, i.e.,
points that do not belong to the original point set $S$.} with diameter $o(\alpha(m,n))$ and $o(m)$ edges.

In this paper we demonstrate that as far as the diameter and number of edges are concerned, 
\emph{Steiner points do not help}. Consequently, the upper bound construction of \cite{CG06}, 
as well as the constructions of Euclidean spanners and spanners for tree metrics that are provided in the current paper,
cannot be improved 
even if one allows the spanner to employ (arbitrarily many) Steiner points. 

\subsection{Our and Previous Techniques} \label{overview}
Arya et al.\ \cite{ADMSS95} demonstrated that it is possible to represent Euclidean $(1+\epsilon)$-spanners as a forest $\mathcal F$ that consists of a constant number of \emph{dumbbell trees}, such that for any pair $p,q$ of distinct points, there exists a dumbbell tree
$T \in \mathcal F$, which satisfies that the path between $p$ and $q$ in $T$ is a $(1+\epsilon)$-spanner path. 
This remarkable property provides a powerful tool for reducing problems on general graphs to similar problems on trees.
Indeed, both constructions of Euclidean sparse spanners with bounded diameter of \cite{ADMSS95,NS07} employ the following four-step scheme. First, build a Euclidean $(1+\epsilon)$-spanner with linear number of edges (and possibly huge diameter) as done, e.g., in \cite{CK93}. Second,
decompose the spanner into a constant number of dumbbell trees as mentioned above. Third, 
build a sparse 1-spanner with bounded diameter for each of these dumbbell trees. 
Finally,
 return the 
union of these 1-spanners as the ultimate spanner. 
Chan and Gupta \cite{CG06} employ a similar approach for building their spanners for doubling metrics. Roughly speaking, instead of working with \emph{dumbbell trees}, Chan and Gupta use 
\emph{net trees} 
that share similar properties.

Our construction of Euclidean spanners also follows the above four-step scheme. 
\ignore{In other words,
to produce our ultimate construction of Euclidean spanners, we devise a construction of 1-spanners for tree metrics and plug it on top of the dumbbell trees of \cite{ADMSS95}. 
We recall that for each parameter $k \ge 2$, our construction of 1-spanners for tree metrics should produce 
within $O(n \alpha_k(n))$ time a
1-spanner with diameter at most $k$ and $O(n \alpha_k(n))$ edges. 
}
Next, we discuss the technical challenges we faced on the way to achieving an optimal-time construction
of 1-spanners for arbitrary $n$-point tree metrics with diameter at most $k$ and $O(n \alpha_k(n))$ edges,
where $n \ge 0$ and $k \ge 2$ are arbitrary integers. 


A central ingredient in the constructions of 1-spanners for tree metrics of \cite{Chaz87,AS87,BTS94,NS07}
is a tree decomposition procedure. Given an $n$-vertex rooted tree $(T,rt)$ and a parameter $\ell$, this procedure
computes a set $CV_\ell$ of at most $O(n/\ell)$ cut-vertices whose removal from the tree
decomposes $T$ into a collection of subtrees of size at most $\ell$ each. 
For our purposes, it is crucial that the running time of this procedure will be $O(n)$. Equally important, the 
size of the set $CV_\ell$ must not exceed $n/\ell$. None of the decomposition procedures of \cite{Chaz87,AS87,BTS94,NS07}
satisfies these two requirements simultaneously. 
The decomposition procedure of \cite{NS07}, for example, requires
time $O(n \log n)$ rather than $O(n)$, and the size bound of $CV_\ell$ is $2(n/\ell)$ rather than $n/\ell$.  We remark that the slack of two on the size bound of $CV_\ell$ contributes a factor of $2^k$ to the number of edges and the running time of the spanner construction
of \cite{NS07}. Also, the slack of $\log n$ on the running time of this  procedure contributes an additional factor of $\log n$ to the the
running time of the construction of \cite{NS07}.
Consequently, the  number of edges and the running time of the spanner construction of \cite{NS07} are bounded above by $O(n 2^k \alpha_k(n))$
and  $O(n (\log n) 2^k \alpha_k(n))$, respectively.
The decomposition procedure
of \cite{BTS94} is the only one in which the size bound of $CV_\ell$ is no greater than $n/\ell$, but it is unclear whether this procedure can be implemented
efficiently. In this paper we provide a decomposition procedure that satisfies both these requirements. Our procedure
is, in addition, surprisingly simple. 

Special attention should be given to determining an optimal value for the parameter $\ell$. In particular, 
the value of $\ell$ was set to be $\alpha_{k-2}(n)$ in both \cite{AS87} and \cite{NS07}. 
In this paper we define a variant $\alpha'_k$ of the function $\alpha_k$, such that $\alpha_k(n) \le \alpha'_k(n) \le 2\alpha_k(n)+4$,
for all $k,n \ge 0$, 
and demonstrate that $\alpha'_{k-2}(n)$ is a much better choice for $\ell$ than $\alpha_{k-2}(n)$ is.
(See Section \ref{s:value} for the formal definitions of the functions $\alpha_k$ and $\alpha'_k$.) In particular,
this optimization 
enables us to ``shave'' a factor of $k$ from both the number of edges and the running time of our construction,
thus proving that this choice of $\alpha'_{k-2}(n)$ for $\ell$ is, in fact, optimal.

Another key ingredient in the constructions of \cite{Chaz87,AS87,BTS94,NS07} is the computation of
an edge set $E'$ that connects the cut vertices of $CV_\ell$. 
Alon and Schieber \cite{AS87} and Narasimhan and Smid \cite{NS07} 
employed a natural yet inherently suboptimal approach.
First, construct a tree $T'$
on the vertex set $CV_\ell \cup \{rt\}$ that ``inherits''
 the tree structure of the original tree $(T,rt)$, by making each vertex $u$ of $CV_\ell$ a child of the first vertex of $CV_\ell \cup \{rt\}$ on the path in $T$
from $u$ to $rt$. Second, recursively compute a sparse 1-spanner for $T'$ with diameter at most $k$.
Note that every 1-spanner path for $T'$ between a pair $u,v$ of vertices, such that $u$ is
an ancestor of $v$ in $T'$, is also a 1-spanner path for $T$.
However, this property does not necessarily hold for a general pair of vertices in $T'$, since
their \emph{least common ancestor} might not be in $T'$.
Consequently, a 1-spanner for $T'$ with diameter at most $k$ provides a 1-spanner for $T$
with diameter at most $2k$ rather than $k$. (See Chapter 12 in \cite{NS07} for further detail.)
To overcome this obstacle, Chazelle \cite{Chaz87} suggested connecting the vertices of $CV_\ell$
into a \emph{Steiner tree} $\mathcal T^*$
using as many additional \emph{Steiner vertices} as needed to guarantee that every 1-spanner path for $\mathcal T^*$ will also be a 1-spanner path for the
original tree $T$. 
\ignore{Chazelle \cite{Chaz87} observed that one can do just fine with only $|CV_\ell|$ Steiner vertices at hand, and so 
the size of the resulting tree $\tilde T$ need not grow beyond twice the size of $CV_\ell$. However, this small slack
of two itself contributes a factor of $2^k$ to the bounds on the number of edges and the running time of the construction.
Subsequently, }
Bodlaender et al.\ \cite{BTS94} took this idea of \cite{Chaz87} one step further
and studied a generalized problem of constructing 1-spanners for arbitrary \emph{Steiner tree metrics}.\footnote{Bodlaender et al.\ \cite{BTS94} referred to this problem as the \emph{Restricted Bridge Problem}.}  
Specifically, suppose that in $T$, a subset $R(T) \subseteq V(T)$ of the vertices are colored black,
and the rest of the vertices in $S(T) = V(T) \setminus R(T)$ are colored white.
The black (respectively, white) vertices are called the \emph{required vertices} (resp., \emph{Steiner
vertices}) of $T$.     
We say that a 1-spanner $H$ for $T$ has diameter at most $k$ if
$H$ contains a 1-spanner path for $T$ 
that consists of at most $k$ edges, for every pair
of \emph{required} vertices in $T$. Bodlaender et al.\ \cite{BTS94} provided a construction
of 1-spanners
for arbitrary Steiner tree metrics with diameter at most $k$ and $O(n \alpha_k(n))$ edges, for any \emph{constant} $k \ge 4$.
However, it was 
unclear prior to this work whether this construction of \cite{BTS94} can be implemented in subquadratic time \cite{NS07}.
In this paper we combine some ideas of \cite{Chaz87,AS87,BTS94,NS07} with numerous
new ideas to produce an algorithm that implements 
the construction of Bodlaender et al.\ \cite{BTS94} in \emph{optimal-time}, and, in addition, extends it
to super-constant values of $k$.
In particular, we devise a linear time procedure for \emph{pruning} the \emph{redundant} vertices of a Steiner tree,
while preserving its basic structure  and intrinsic properties.
Our algorithm makes an extensive use of this pruning procedure, 
e.g., for pruning the initial Steiner tree from its redundant vertices, and for computing the edge set $E'$
that connects the cut vertices of $CV_\ell$.



Finally, our extension of the lower bound of \cite{CG06} to Euclidean Steiner spanners employs a direct combinatorial argument,
which shows that any Steiner spanner can be ``pruned'' from Steiner points, while increasing the number of edges by a small factor
and preserving the same diameter.
More specifically, we demonstrate that it is possible to replace each edge of the original Steiner spanner with
a constant number of edges--none of which is incident on a Steiner point, so that for every pair of required points and 
any path $P$ in the original Steiner spanner between them, there would be a path $P'$ in the resulting graph between
these points that
consists of the same number of edges as $P$ and 
whose weight is no greater than that of $P$.

\subsection{Structure of the Paper}
In Section \ref{s:value} we present some very slowly growing functions that are used throughout
the paper, and analyze their properties. The technical proofs involved in this analysis are relegated to Appendices \ref{acker} and \ref{acker2}.
Section \ref{1span} is devoted to our construction of 1-spanners for tree metrics.
Therein we start (Section \ref{s:scheme}) with outlining our basic scheme. We proceed (Section \ref{s:prune}) with
presenting the pruning procedure   
and providing a few useful properties of the resulting \emph{pruned} trees. The decomposition procedure is given in Section \ref{decompose}. Finally, in Section \ref{s:sol} we provide an optimal-time algorithm
for computing 1-spanners for tree metrics and analyze its performance.
In Section \ref{s:upper} we derive our construction of Euclidean spanners. 
Our lower bounds for Euclidean Steiner spanners are established in Section \ref{s:lower}. 

\subsection{Definitions and Notation}
The \emph{size} of a tree $T$,
denoted $|T|$, is the number of vertices in $T$.
The number of edges in a path $P$ is denoted by $|P|$, and the weight of $P$ is denoted by $w(P)$.
For a  tree $T$ and a subset $C$ of $V(T)$, we denote by $T
\setminus C$ the forest obtained from $T$ by removing all
vertices in $C$ along with the edges that are incident to
them.
For a positive integer $n$, we denote the set $\{1,2,\ldots,n\}$ by
$[n]$. In what follows all logarithms are in base 2.


\ignore{
\subsection{Preliminaries}
For a rooted tree $(T,rt)$ and a subset $C$ of $V(T)$, we denote by $T \setminus C$ the forest obtained from $T$ by removing all the vertices of $C$ together
with their incident edges.

Weight of a path, length of a path, distance of a path...
}

\section{Some Very Slowly Growing Functions} \label{s:value}
In this section we present a number of very slowly growing functions that are used throughout. 

Following \cite{Tarj75,AS87,NS07}, we define the following two very rapidly growing functions $A_k$ and $B_k$:
\begin{itemize}
\item $A_k(n) = A_{k-1}(A_k(n-1))$, for all $k,n \ge 1$;
$A_0(n) = 2n$, for all $n \ge 0$; $A_k(0) = 1$, for all $k \ge 1$.
\item $B_k(n) = B_{k-1}(B_k(n-1))$, for all $k,n \ge 1$;
$B_0(n) = n^2$, for all $n \ge 0$; $B_k(0) = 2$, for all $k \ge 1$.
\end{itemize}
We define the functional inverses of the functions $A_k$ and $B_k$ in the following way:
\begin{itemize}
\item $\alpha_{2k}(n) = \min\{s \ge 0: A_k(s) \ge n\}$, for all $k,n \ge 0$.
\item $\alpha_{2k+1}(n) = \min\{s \ge 0: B_k(s) \ge n\}$, for all $k,n \ge 0$.
\end{itemize}

For technical convenience, we define $\log 0 = 0$.
Observe that for all $n \ge 0$:
$\alpha_0(n) = \lceil n/2 \rceil$,
$\alpha_1(n) = \left\lceil \sqrt{n} \right\rceil$,
$\alpha_2(n) = \lceil \log{n} \rceil$,
$\alpha_3(n) = \lceil \log\log{n} \rceil$,
$\alpha_4(n) = \log^* n$,
$\alpha_5(n) = \lfloor \frac{1}{2}\log^* n \rfloor$, \ldots, etc.

The following lemma can be easily verified. 
\begin{lemma} \label{basereal}
(1) For all $k \ge 0$,
the function $\alpha_k = \alpha_k(n)$ is monotone non-decreasing with $n$.
~(2) For all $k \ge 2$ and $n \ge 1$, $\alpha_{k}(n) < n$. Also, for all $n \ge 2$ (respectively, $n \ge 3$),
we have $\alpha_0(n) < n$ (resp., $\alpha_1(n) < n$).
~(3) For all $k \ge 0$ and $n \ge 0$, $\alpha_{k+2}(n) \le \alpha_{k}(n)$.
\end{lemma} 

The following lemma 
from \cite{NS07} provides a useful characterization of the function $\alpha_{k}$.
\begin{lemma} [Lemma 12.1.16 in \cite{NS07},  p.\ 230] \label{NSreal}
Let $k \ge 2$ be an arbitrary integer. Then
$\alpha_{k}(n) = 1+ \alpha_{k}(\alpha_{k-2}(n))$, for all $n \ge 3$, and
 $\alpha_{k}(0) = \alpha_{k}(1) = 0$. Also, $\alpha_k(2) = 0$ if $k$ is odd, and $\alpha_k(2) = 1$ if $k$ is even. 
\end{lemma}


Next, we define a variant $\alpha'_k(n)$ of the function $\alpha_{k}$.
\begin{itemize}
\item 
$\alpha'_{0}(n) = \alpha_0(n) = \lceil n/2 \rceil$, for all $n \ge 0$;
$\alpha'_{1}(n) = \alpha_1(n) = \left\lceil \sqrt{n} \right\rceil$, for all $n \ge 0$
\item $\alpha'_{k}(n) = 
2+\alpha'_{k}(\alpha'_{k-2}(n))$, for all $k \ge 2$ and $n \ge k+2$;
\\$\alpha'_{k}(n) = \alpha_{k}(n)$, for all $k \ge 2$ and $n \le k+1$.
\end{itemize}

The following lemma, whose proof appears in Appendix \ref{acker}, is analogous to Lemma \ref{basereal}.
It establishes key properties of the function $\alpha'_k$ that will be used in the sequel.
\begin{lemma} \label{base'real}
(1) For all $k \ge 2$, the function $\alpha'_k = \alpha'_k(n)$ is monotone non-decreasing with $n$.
~(2) For all $k \ge 2$ and $n \ge 1$, $\alpha'_{k}(n) < n$. Also, for all $n \ge 2$ (respectively, $n \ge 3$),
we have $\alpha'_0(n) < n$ (resp., $\alpha'_1(n) < n$).
~(3) For all $k \ge 2$ and $n \ge 0$, $\alpha'_{k+2}(n) \le \alpha'_{k}(n)$.
\end{lemma} 

Observe that for all $k \ge 0$ and $n \ge 0$, $\alpha'_k(n) \ge \alpha_k(n)$.
The following lemma, whose proof appears in Appendix \ref{acker2}, shows that $\alpha'_k(n)$ is not much greater than $\alpha_k(n)$.
\begin{lemma} \label{finallyreal}
For all $k,n \ge 0$, 
$\alpha'_{k}(n) \le  2 \alpha_k(n) + 4$. 
\end{lemma}

The \emph{Ackermann function} is defined by $A(n) = A_n(n)$, for all $n \ge 0$,
and the one-parameter \emph{inverse-Ackermann function} is defined by $\alpha(n) = \min\{s \ge 0:A(s) \ge n\}$,
for all $n \ge 0$.
In \cite{NS07} it is shown that
 $\alpha_{2\alpha(n) + 2}(n) \le 4$.
(A similar bound was established in \cite{LaPo90}.)
By Lemma \ref{finallyreal}, we get that $\alpha'_{2\alpha(n)+2}(n) \le 12$.
Finally, the two-parameter \emph{inverse Ackermann function} is defined by $\alpha(m,n) = \min\{s \ge 1: A(s,4\lceil m/n \rceil) \ge \log n\}$, for all $m,n \ge 0$. \section{1-Spanners for Tree Metrics} \label{1span}
In this section we present our construction of 1-spanners for tree metrics.

\subsection{The Basic Scheme} \label{s:scheme}
Let $(T,rt) = (V,E,w)$ be an arbitrary $n$-vertex weighted rooted tree, and let $M_T$
be the tree metric induced by $T$. Our goal is to compute a sparse 1-spanner for
$T$ with bounded diameter. Clearly, $T$ is already a sparse 1-spanner for itself, but its diameter may be huge.
We would like to reduce the diameter of $T$ by adding to it a small number of edges.

For a pair $u,v$ of vertices in $T$, we denote by $P_T(u,v)$ the unique path in $T$ between $u$ and $v$.
Let $H$ be an arbitrary \emph{unweighted} graph on the vertex set $V$ of $T$.   
A path $P$ in $H$ between $u$ and $v$ is called $T$-monotone if it is a sub-path of $P_T(u,v)$, i.e.,
if we write $P_T(u,v) = (u = v_0, v_1, \ldots, v = v_t)$, then $P$ can be written as $P = (u = v_{i_0}, v_{i_1}, \ldots,
v_{i_q})$, where $0 = i_0 < i_1 < \ldots < i_q = t$. 
The \emph{$T$-monotone distance} between a pair $u,v$ of vertices in $H$ is defined as the minimum
number of edges in a $T$-monotone path in $H$ connecting them. The \emph{$T$-monotone diameter}
of $H$, denoted $\Lambda_T(H) = \Lambda(H)$, is defined as the maximum $T$-monotone distance between any pair of vertices in $H$.
(If $T$ is clear from the context, we may write diameter
instead of $T$-monotone diameter.)
By the triangle inequality, for any $T$-monotone path in $H$, the corresponding weighted path in $M_T$ provides a 1-spanner
path for $T$. Hence, $H$ translates into a 1-spanner for $T$ with diameter $\Lambda(H)$, and this holds 
true regardless of the actual weight function $w$ of $T$. 
We henceforth restrict attention to unweighted trees in the sequel.

Following \cite{BTS94}, we study a generalization of the problem for \emph{Steiner trees}, where
there is a designated subset $R(T) \subseteq V$ of \emph{required vertices}, and the diameter of a 1-spanner 
for $T$ is defined as the maximum $T$-monotone distance between any pair of \emph{required vertices}. 
The \emph{required-size} of a Steiner tree is defined as the number of required vertices in it.
Also, the remaining vertices in $S(T) = V \setminus R(T)$ are called the \emph{Steiner vertices} of $T$.
If the number
of Steiner vertices in $T$ is (much) larger than the number of required
vertices, it might be possible to \emph{prune} some \emph{redundant} Steiner vertices from $T$
while preserving its basic structure and intrinsic properties.
A Steiner rooted tree $(T',rt')$ is called \emph{$T$-monotone preserving}, if
(1) $R(T') = R(T)$, and (2)
for every pair $u,v$ of required vertices, the unique path $P_{T'}(u,v)$ between $u$ and $v$ in $T'$ is
$T$-monotone. 
Consider a $T$-monotone preserving tree $(T',rt')$, and let
$u,v$ be an arbitrary pair of required vertices. Note that any $T'$-monotone path
between $u$ and $v$ is also $T$-monotone.
Thus any 1-spanner $H'$ for $T'$ with $T'$-monotone diameter $k$
is also a 1-spanner for $T$ with $T$-monotone diameter $k$. 

Our algorithm for constructing 1-spanners for Steiner tree metrics employs the following recursive scheme.
We start by \emph{pruning} the redundant vertices of $T$, thus transforming $T$ into a $T$-monotone
preserving tree $T'$ that does not contain too many 
Steiner vertices. We then select a set $CV_\ell$ of at most $n/\ell$ cut-vertices whose removal from the
tree decomposes it into a collection of subtrees of required-size at most $\ell$ each, for some parameter
 $\ell$.
Next, we would like to connect the cut-vertices using a small number of edges, so that the $T$-monotone
distance between any pair of cut-vertices will be small. To this end we (1) compute a copy $\tau$ of $T$ in
which the designated set of required vertices is $CV_\ell$, (2) prune $\tau$ from
its redundant vertices, and (3) call the algorithm recursively on the resulting pruned tree.
We then add a small number of edges to connect between cut vertices and subtrees in the spanner. This step is simple and does not involve a recursive call of the algorithm. Finally, we prune each of the subtrees from redundant vertices, 
and then call the algorithm recursively for each of them.

\ignore{
Our solution to Problem \ref{shortcuttingst} starts (Section \ref{s:prune}) with presenting a linear time
procedure for removing the redundant Steiner vertices of $T$, thus transforming $T$ into a $T$-monotone
preserving tree $T'$ that does not contain too many 
Steiner vertices. 
We proceed  with proving a number of key properties of the resulting \emph{pruned} tree $T'$. 
In particular, we show therein that the number of Steiner vertices of $T'$ is smaller than the number of required vertices.
Next, we present (Section \ref{decompose}) a linear time procedure that, given a Steiner tree $T$, computes a set ${CV}_\ell$ of at most $n/\ell$ cut vertices 
whose removal from $T$ partitions it into subtrees of required-size at most $\ell$ each. 
Determining an optimizing value for $\ell$ plays a major role in most of the previous solutions ...]
to  Problem \ref{shortcutting}, as well as in the current solution.
In \cite{AS87} and \cite{NS07}, for example, $\ell$ was set to be $\alpha_{k-2}(n)$, where $n$ is the size of the tree,
$k$ is the desired bound on the diameter, and $\alpha_k(n)$ is the inverse of a certain function at the
$\lceil k/2 \rceil$th level of the primitive recursive hierarchy. 
In other works ] $\ell$ was set to some value very close to $\alpha_{k-2}(n)$.
In Section \ref{s:value} we define a variant $\alpha'_k(n)$ of the function $\alpha_k(n)$, 
and later demonstrate that 
the optimizing value for $\ell$ should be $\alpha'_{k-2}(n)$ rather than $\alpha_{k-2}(n)$. S: very bad]]
Finally, in Section \ref{s:sol} we present an optimal solution to Problem \ref{shortcuttingst}. 

A central ingredient of previous constructions of 1-spanners for tree metrics ...] 
is the computation
of a set $CV_\ell$ of $O(|T|/\ell)$ cut vertices whose removal from $T$
partitions it into subtrees of size $O(\ell)$ each, where $\ell$ is some parameter to be discussed later.
Having computed $CV_\ell$, it is required to 
connect the cut vertices using a minimal number of edges, 
so that the monotone distance between any pair of cut vertices will be small. 
To this end, Bodlaender et al.\ \cite{BTS94} suggested to study 
the following generalization 
of Problem 
\ref{shortcutting} 
for \emph{colored} or \emph{Steiner} trees, in which we no longer
need to bound the hop-distances for \emph{all} vertices, but
rather only for vertices that belong to a \emph{designated
subset} of $V$, e.g., $CV_\ell$.
Specifically, suppose that in $T$, a subset $R(T) \subseteq V$ of the vertices are colored black,
and the rest of the vertices in $S(T) = V \setminus R(T)$ are colored white.
The black (respectively, white) vertices are called the \emph{required vertices} (resp., \emph{Steiner
vertices}) of $T$. The number $|R(T)|$ of required vertices in $T$ is called the \emph{required-size}
of $T$.
Let $H$ be a spanning subgraph of $M_T$.
The \emph{$T$-monotone diameter} (or \emph{hop-diameter}) S: or \emph{diameter}]] of $H$ is defined
to  be the maximum hop-distance between any two \emph{required} vertices in $H$.
\begin{problem}  \label{shortcuttingst}
Given a weighted Steiner tree $T$ and an integer $k \ge 2$, 
efficiently compute a sparsest possible
1-spanner $H$ for $T$, $E(T) \subseteq E(H)$, with hop-diameter at most $k$.
\end{problem}
}
\subsection{The Pruning Procedure} \label{s:prune}
In this section we devise a procedure for pruning the redundant
vertices of a Steiner tree 
while preserving its basic structure. 
In addition, we provide a few useful properties of pruned trees.

For a Steiner rooted tree $(T,rt)$ and a pair $u,v$ of vertices in $T$,
we denote by $LCA_T(u,v)$ the \emph{least common ancestor} (henceforth, LCA)
of $u$ and $v$ in $T$.
A Steiner vertex $x \in S = S(T)$ in $T$
is called \emph{useful} if it is the LCA
of some pair of required vertices $u,v \in R = R(T)$. Otherwise it is called \emph{redundant}.
We denote by $LCA(T)$ the set of all 
useful vertices in $T$, i.e., $LCA(T) =  \{x \in S ~\vert~ \exists u,v \in R: x = LCA_T(u,v)\}$.
A Steiner rooted tree with no redundant vertices is called \emph{pruned}.

We denote the children of the root vertex $rt$ in a Steiner rooted tree $(T,rt)$ by $c_1,\ldots,c_{ch(rt)}$, where $ch(rt)$ denotes
the number of children of $rt$ in $T$. For each index $i \in [ch(rt)]$, let $T_{(i)}$ be the subtree of $T$
rooted at $c_i$. We say that the subtree $T_{(i)}$ is \emph{required} if it contains at least one
required vertex, i.e., if $R_{(i)} = R \cap V(T_{(i)})$ is non-empty. Otherwise we say that it is \emph{redundant}.
Notice that all vertices in a redundant subtree are redundant.
Denote by $I = I(T)$ the set of all indices $i$, such that $i \in [ch(rt)]$ and $T_{(i)}$ is a required subtree.

Next, we present a linear time procedure $Prune = Prune((T,rt))$
that accepts as input a Steiner rooted tree $(T,rt)$,
and transforms it into a pruned $T$-monotone preserving tree $(T',rt')$.

If $T$ consists of just the single vertex $rt$, then the procedure either leaves $T$ intact if $rt \in R$,
or it transforms $T$ into an empty tree if $rt \nin R$.
Otherwise, $|T| \ge 2$.
For each index $i \in [ch(rt)]$, the tree $(T_{(i)},c_i)$ is recursively transformed into a pruned $T_{(i)}$-monotone preserving tree $(T'_{(i)},c'_i)$.
Observe that for each index $i \in [ch(rt)] \setminus I$, $R_{(i)} = \emptyset$ and $T_{(i)}$ is a redundant subtree, 
and so $T'_{(i)}$ is empty. Also, for each index $i \in I$,
the subtree $T'_{(i)}$ is non-empty. 
The procedure removes all $[ch(rt)]$ edges connecting the root vertex $rt$ of $T$ with its children. 
The execution of the procedure then splits into four cases.
\\\emph{Case 1: $rt \in R$.} The root vertex $rt$ of $T$ remains the root vertex of $T'$, and for each
index $i \in [I]$, an edge connecting $rt$ with the root $c'_i$ of $T'_{(i)}$ is added.
\\\emph{Case 2: $rt \nin R$ and $I = \emptyset$.} Hence, $R = \emptyset$, and the procedure transforms $T$ into an
empty tree.
\\\emph{Case 3: $rt \nin R$ and $|I| = 1$.} In this case $rt$ is redundant, and there
is a single non-empty subtree $T'_{(p)}$, i.e., $I = \{p\}$, for some index $p \in [ch(rt)]$.
Hence, the procedure removes $rt$ and sets $T' = T'_{(p)}$.
\\\emph{Case 4: $rt \nin R$ and $|I| \ge 2$.} In this case $rt$ is useful. As in
case 1, the root $rt$ of $T$ remains the root vertex of $T'$, and for each
index $i \in [I]$, an edge connecting $rt$ with the root $c'_i$ of $T'_{(i)}$ is added.
\\(See Figure \ref{prun} for an illustration.)

\begin{figure*}[htp]
\begin{center}
\begin{minipage}{\textwidth} 
\begin{center}
\setlength{\epsfxsize}{6.8in} \epsfbox{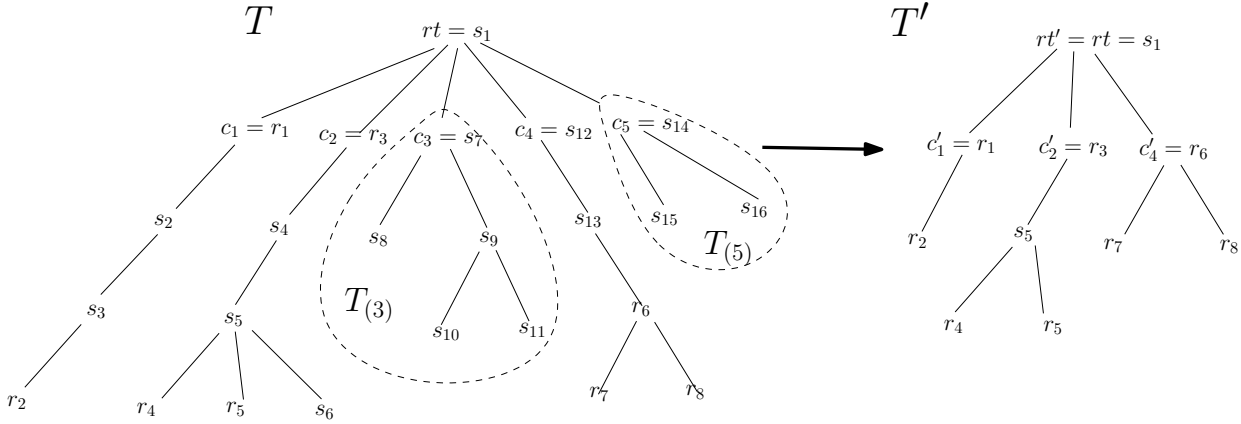}
\end{center}
\end{minipage}
\caption[]{ \label{prun} \sf \footnotesize A rooted Steiner tree $(T,rt)$
is depicted on the left, having 
8 required vertices $r_1,\ldots,r_8$ and 16
Steiner vertices $rt = s_1,\ldots,s_{16}$. The two subtrees
$T_{(3)}$ and $T_{(5)}$ of $T$ that are depicted within dashed lines are redundant, whereas the
other three subtrees $T_{(1)},T_{(2)}$, and $T_{(4)}$ of $T$ are required.
The pruned $T$-monotone preserving tree $(T',rt')$ that is depicted on the right
is obtained as a result of the invocation of the procedure $Prune$ on $T$.}
\end{center}
\end{figure*}

It is easy to verify that the procedure $Prune$ can be implemented in linear time.

Next, we analyze the properties of the resulting tree $T'$.

The following lemma follows easily from the description of the procedure. 
\begin{lemma} \label{description}
$(T',rt')$ is a Steiner rooted tree over $V(T') = R(T) \cup LCA(T)$, and $R(T') = R(T)$.
Also, for each index $i \in I$, $(T'_{(i)},c'_i)$ is a Steiner rooted tree over $V(T'_{(i)}) = R(T_{(i)}) \cup LCA(T_{(i)})$,
and $R(T'_{(i)}) = R(T_{(i)})$.
\end{lemma}

\begin{lemma} \label{hierarchy}
For any pair $u,v$ of vertices in $T'$, $u$ is an ancestor of $v$ in $T'$
iff it is its ancestor in $T$.
\end{lemma}
\begin{proof}
The proof is by induction on $n' = |T'|$.
The basis $n' \le 1$ holds vacuously.
\\\emph{Induction Step:} We assume the correctness of the statement for all smaller values of $n'$, $n' \ge 2$,
and prove it for $n'$. Since $n' \ge 2$, it must hold that $|I| \ge 1$. Next, we prove the ``only if'' part.
The argument for the ``if'' part is similar. 
Consider an arbitrary pair $u,v$ of vertices in $T'$, such that $u$ is an ancestor of $v$ in $T'$.
Next, we show that $u$ is also an ancestor of $v$ in $T$.
By Lemma \ref{description}, for each index $i \in I$, $V(T'_{(i)}) \subseteq V(T_{(i)})$.
The analysis splits into two cases.
\\\emph{Case 1: $|I| =1$ and $rt \nin R$}. In this case $T' = T'_{(p)}$, with $I = \{p\}$.
Notice that $u$ and $v$ belong to $T_{(p)}$.
By the induction hypothesis for $T'_{(p)}$, $u$ is an ancestor of $v$ in $T_{(p)}$, and thus also in $T$.
\\\emph{Case 2: Either $rt \in R$ or $|I| \ge 2$}. 
In both cases $rt(T') = rt$, and for each index $i \in I$,
the root $c'_i$ of the subtree $T'_{(i)}$ is a child of $rt$ in $T'$.
\\If both $u$ and $v$ belong to the same subtree $T'_{(i)}$, for 
some index $i \in I$, then they both belong to $T_{(i)}$. Hence, by the induction hypothesis for $T'_{(i)}$, $u$ 
is an ancestor of $v$ in $T_{(i)}$, and thus also in $T$.
\\Since $u$ is an ancestor of $v$ in $T'$, $u$ and $v$ cannot belong to different subtrees $T'_{(j)}$ and $T'_{(k)}$ of $T'$,
$j,k \in I$.  
Hence, the remaining case is that $u = rt(T') = rt$. Clearly, $rt$ is an ancestor of
$v$ in $T$, and we are done.\QED 
\end{proof}

\begin{lemma} \label{lca}
For any pair $u,v$ of required vertices, $LCA_{T'}(u,v) = LCA_T(u,v)$.
\end{lemma}
\begin{proof}
Write $l' = LCA_{T'}(u,v)$ and $l = LCA_{T}(u,v)$.
First, notice that $l$ is either a required vertex or a useful vertex. By Lemma \ref{description}, we get that $l$ belongs to $T'$.
By definition, $l'$ is the LCA of $u$ and $v$ in $T'$.
By Lemma \ref{hierarchy}, it follows that $l'$ is a common ancestor of $u$ and $v$ in $T$,
and so it must be an ancestor of their LCA $l$ in $T$.
Lemma \ref{hierarchy} implies that $l'$ is an ancestor of $l$ also in $T'$.
However, by applying Lemma \ref{hierarchy} again, we get that $l$ is a common ancestor of $u$ and $v$ in $T'$,
and so it must be an ancestor of their LCA $l'$ in $T'$. It follows that $l' = l$.\QED
\end{proof}

Lemmas \ref{description} and \ref{lca} yield the following corollary.
\begin{corollary}
$(T',rt')$ is pruned. 
\end{corollary}
\begin{proof}
We argue that $LCA(T') = LCA(T)$.
Indeed, by Lemma \ref{description}, $V(T') = R(T) \cup LCA(T)$ and
$R(T') = R(T)$. Hence, $S(T') = LCA(T)$, and so $LCA(T') \subseteq S(T') = LCA(T)$.
To see why $LCA(T) \subseteq LCA(T')$ holds true as well, consider a vertex $l \in LCA(T)$. By definition, $l \nin R(T)$, and there exists a pair of required vertices 
$u,v \in R(T)$, such that $l = LCA_T(u,v)$. Hence, $l \nin R(T')$, and by Lemma \ref{lca}, $l = LCA_{T'}(u,v)$.
It follows that $l \in LCA(T')$.

Consequently, $V(T') = R(T') \cup LCA(T')$, and so there are no redundant vertices in $T'$. \QED
\end{proof}

\begin{lemma} \label{sub-p}
For any pair $u,v$ of vertices in $T'$, such that $u$ is an ancestor of $v$ in $T'$,
$P_{T'}(u,v)$ is $T$-monotone.
\end{lemma}
\begin{proof}
Write $P_{T'}(u,v) = (u=v_0,v_1,\ldots,v=v_q)$. By Lemma \ref{hierarchy}, for each index $i \in [q]$,
$v_{i-1}$ is an ancestor of $v_i$ in $T$. Hence, $P_{T'}(u,v)$ is a sub-path of $P_T(u,v)$, i.e.,
it is $T$-monotone. \QED
\end{proof}
We conclude that $T'$ is $T$-monotone preserving.
\begin{corollary}
For any pair $u,v$ of required vertices, $P_{T'}(u,v)$ is $T$-monotone.
\end{corollary}
\begin{proof}
If $u$ is either an ancestor or a descendant of $v$ in $T'$, then the statement follows from Lemma \ref{sub-p}.

We may henceforth assume that $LCA_{T'}(u,v) \ne u,v$. Write $l' = LCA_{T'}(u,v)$ and $l = LCA_{T}(u,v)$.
By Lemma \ref{lca}, $l' = l$. Observe that $P_{T'}(u,v)$ is a concatenation of the two paths
$P_{T'}(u,l)$ and $P_{T'}(l,v)$, i.e., 
$P_{T'}(u,v) = P_{T'}(u,l) \circ P_{T'}(l,v)$.
Similarly, we have $P_T(u,v) = P_T(u,l) \circ P_T(l,v)$.
Lemma \ref{sub-p} implies that both $P_{T'}(u,l)$ and $P_{T'}(l,v)$ are $T$-monotone,
or equivalently, $P_{T'}(u,l)$
is a sub-path of $P_{T}(u,l)$ 
and $P_{T'}(l,v)$ is a sub-path of $P_{T}(l,v)$.
It follows that $P_{T'}(u,v) = P_{T'}(u,l) \circ P_{T'}(l,v)$ is a sub-path of $P_{T}(u,v)
= P_{T}(u,l) \circ P_{T}(l,v)$, i.e.,
$P_{T'}(u,v)$ is $T$-monotone. \QED
\end{proof}

Having proved that $T'$ is a pruned $T$-monotone preserving tree, we now turn to establish a number of basic properties of pruned trees that will be of use in
the sequel.

A Steiner tree in which the number of Steiner vertices is smaller than the number of required vertices is called \emph{compact}.
Note that in any (non-empty) pruned tree $T$, $R \ne \emptyset$ and $S = LCA(T)$.
The next lemma implies that any non-empty pruned tree is compact.
\begin{lemma} \label{compa}
For any Steiner rooted tree $(T,rt)$ (not necessarily pruned), 
$|LCA(T)| \le \max\{0,|R|-1\}$.
\end{lemma}
\begin{proof}
The proof is by induction on $n = |T|$. The basis $n=0$ is trivial.
\\\emph{Induction Step:} We assume the correctness of the statement for all smaller values of $n$, $n \ge 1$,
and prove it for $n$. 
If $R = \emptyset$, then by definition $LCA(T) = \emptyset$
as well, and we are done. \\We henceforth assume that $R$ is non-empty, and so $\max\{0,|R|-1\} = |R|-1$.
By definition, for each index $i \in I$, $R_{(i)} \ne \emptyset$, and for each index $i \in [ch(rt)] \setminus I$,
$R_{(i)} = \emptyset$.
Hence, by the induction hypothesis, for each index $i \in I$, $|LCA(T_{(i)})| \le  \max\{0,|R_{(i)}|-1\} = |R_{(i)}|-1$,
and for each index $i \in [ch(rt)] \setminus I$, $LCA(T_{(i)}) = \emptyset$. 
Clearly, the sets $\{R_{(i)}\}_{i \in I}$ and $\{LCA(T_{(i)})\}_{i \in I}$ are pairwise disjoint.
The analysis splits into three cases.
\\\emph{Case 1: $rt \in R$.} 
In this case $R = \bigcup_{i\in I} R_{(i)} \cup \{rt\}$ and $LCA(T) = \bigcup_{i \in I} LCA(T_{(i)})$, implying
that $|R| = \sum_{i\in I} |R_{(i)}| + 1$ and $|LCA(T)| = \sum_{i \in I}  |LCA(T_{(i)})|$.
Altogether, \begin{eqnarray*}
|LCA(T)| &=& \sum_{i \in I} |LCA(T_{(i)})| ~\le~ \sum_{i \in I} \left(|R_{(i)}| -1\right)
~=~ \sum_{i \in I} |R_{(i)}| -|I| ~\le~ \sum_{i=1} |R_{(i)}| ~=~ |R|-1.
\end{eqnarray*} 
\\\emph{Case 2: $rt$ is redundant, i.e., $rt \in S \setminus LCA(T)$.}
Since $R \ne \emptyset$ and $rt$ is redundant, it must hold that $|I| = 1$, i.e., $I = \{p\}$,
for some index $p \in [ch(rt)]$.
Hence, $R = R_{(p)}$ and $LCA(T) = LCA(T_{(p)})$, implying that
$|LCA(T)| ~=~ |LCA(T_{(p)})| ~\le~  |R_{(p)}| -1 ~=~ |R|-1.$
\\\emph{Case 3: $rt$ is useful, i.e., $rt \in LCA(T)$.}
In this case there must be at least two different required subtrees $T_{(j)}$ and $T_{(k)}$, $j,k \in I$, and 
 so $|I| \ge 2$.
Observe that $R = \bigcup_{i \in I} R_{(i)}$ and $LCA(T) = \bigcup_{i \in I} LCA(T_{(i)}) \cup \{rt\}$, implying
that $|R| = \sum_{i \in I} |R_{(i)}|$ and $|LCA(T)| = \sum_{i \in I} |LCA(T_{(i)})|+1$. 
It follows that
\begin{eqnarray*}
|LCA(T)| &=& \sum_{i \in I} |LCA(T_{(i)})| + 1 ~\le~ \sum_{i \in I} \left(|R_{(i)}| -1\right) + 1
~=~ \sum_{i\in I} |R_{(i)}| -|I| +1 \\ &\le& \sum_{i \in I} |R_{(i)}| -1 ~=~ |R|-1.
\end{eqnarray*}
\QED
\end{proof}

\begin{lemma} \label{hop}
For a non-empty pruned tree $(T,rt)$, its depth $h(T)$ is at most $|R|-1$ and its diameter $\Lambda(T)$ is at most $|R|$. 
Moreover, $\Lambda(T)$ is equal to $|R|$ only if the following conditions
hold: 
(1) $rt \nin R$,  
~(2) $rt$ has exactly two children, and 
~(3) For any pair $u,v$ of vertices in $T$ for which $|P_T(u,v)| = |R|$, 
$u,v \ne rt = LCA_T(u,v)$.
\begin{proof}
The proof is by induction on $n = |T|$. The basis $n=1$ is trivial.
\\\emph{Induction Step:} We assume the correctness of the statement for all smaller values of $n$, $n \ge 2$, and prove it
for $n$. 
Since $T$ is pruned, 
all the subtrees $T_{(1)},\ldots,T_{(ch(rt))}$ of $T$ are pruned as well, and so 
the induction hypothesis applies to every one of them.

Fix an arbitrary index $i \in [ch(rt)]$.
Since $T_{(i)}$ is pruned, we have $|R_{(i)}| \ge 1$. We argue that $|R_{(i)}| \le |R|-1$.
This is clearly the case if $rt \in R$. Otherwise, $rt$ must be useful, and so 
it must have at least two children in $T$.
Hence, there
is another index $j \in [ch(rt)]$, such that $|R_{(j)}| \ge 1$. Since $R_{(i)} \cup R_{(j)} \subseteq R$,
we get that $|R_{(i)}| \le |R| - |R_{(j)}| \le |R|-1$.

By the induction hypothesis, for each index $i \in [ch(rt)]$, $h(T_{(i)}) \le |R_{(i)}|-1 \le |R|-2$.
Hence, $$h(T) ~=~ \max\{h(T_{(i)}) ~\vert~ i \in [ch(rt)]\} + 1 ~\le~ |R|-2 + 1 = |R|-1.$$

To bound the diameter $\Lambda(T)$ of $T$, consider a pair $u,v$ of vertices in $T$
for which $|P_T(u,v)| = \Lambda(T)$. 
If $u$ and $v$ belong to the same subtree $T_{(i)}$ of $T$,
for some index $i \in [ch(rt)]$, then $\Lambda(T) = \Lambda(T_{(i)})$, and by the induction hypothesis
for $T_{(i)}$, we get that $\Lambda(T) = \Lambda(T_{(i)}) \le |R_{(i)}| \le |R|-1$.
Otherwise, $rt = LCA_T(u,v)$. If either $u$ or $v$ is the root vertex $rt$, then
$\Lambda(T) \le h(T) \le |R|-1$. 
\\So far we have proved that in order to obtain $\Lambda(T) \ge |R|$, it must hold that
$u,v \ne rt = LCA_T(u,v)$.
We may henceforth assume that $u,v \ne rt = LCA_T(u,v)$. In other words,
$u$ and $v$ belong to different subtrees $T_{(i)}$ and $T_{(j)}$
of $T$, respectively, for some indices $i,j \in [ch(rt)]$. Observe that $|P_T(u,v)| \le h(T_{(i)}) + h(T_{(i)}) + 2$.
By the induction hypothesis for $T_{(i)}$ and $T_{(j)}$,
$h(T_{(i)}) \le |R_{(i)}| -1$ and $h(T_{(j)}) \le |R_{(j)}| -1$, and so $|P_T(u,v)| \le |R_{(i)}| + |R_{(j)}|$.
It follows that $\Lambda(T) = |P_T(u,v)| \le |R_{(i)}| + |R_{(j)}| \le |R|$.
Moreover, one can have $\Lambda(T) = |R|$ only 
if $|R_{(i)}| + |R_{(j)}| = |R|$, in which case both
$rt \nin R$ and $ch(rt) = 2$ must hold.
\QED
\end{proof}
\end{lemma}
\begin{corollary} \label{corhop}
Let $(T,rt)$ be a pruned tree, such that $rt$ has exactly two children $c_1$ and $c_2$,
and let $\tilde T$ be the graph obtained from $T$ by adding to it the edge $(c_1,c_2)$.
Then the $T$-monotone diameter $\Lambda(\tilde T)$ of $\tilde T$ is at most $|R|-1$.
\end{corollary}
\begin{proof}
Consider a pair $u,v$ of vertices in $\tilde T$
for which their $T$-monotone distance $\delta$
satisfies $\Lambda(\tilde T) = \delta$. 
Since $\tilde T$ contains all edges of $T$, we have $\delta \le |P_{T}(u,v)|$.
If $|P_{T}(u,v)| \le |R|-1$, then we are done. Otherwise,
by Lemma \ref{hop}, $|P_{T}(u,v)| = |R|$ and $u,v \ne rt = LCA_T(u,v)$. 
Hence, either $u$ belongs to $T_{(1)}$  and $v$ belongs to $T_{(2)}$, or vice versa.
Suppose without loss of generality that $u$ belongs to $T_{(1)}$ and $v$ belongs to 
$T_{(2)}$, and write $P_T(u,v) = (u = v_0,v_1,\ldots,c_1 = v_{j-1},rt = v_j,c_2 = v_{j+1},v_{j+2},\ldots,v=v_{|R|})$.
Notice that $\tilde T$ contains all edges of $P_T(u,v)$, and, in addition, the edge $(c_1,c_2)$,
which can be used as a shortcut to avoid the detour $(c_1,rt,c_2)$ around $rt$.
Hence, $\tilde T$ contains the $T$-monotone path
$\tilde P =  (u = v_0,v_1,\ldots,c_1 = v_{j-1},c_2 = v_{j+1},v_{j+2},\ldots,v=v_{|R|})$
that consists of $|R|-1$ edges, and so $\Lambda(\tilde T) = \delta \le |\tilde P| = |R|-1$. \QED
\end{proof}

\subsection{Tree Decomposition Procedure} \label{decompose}
In this section we devise a procedure $Decomp$ for decomposing a Steiner tree into subtrees
in an optimal way.

Let $n$ be an arbitrary positive integer. 
The procedure $Decomp = Decomp((T,rt),\ell)$ accepts as input a Steiner rooted tree $(T,rt)$
with required-size  $n$ 
and a positive integer $\ell$, and returns as output a  set $CV_\ell \subseteq V(T)$ 
of \emph{cut vertices}.
We do not require that 
a cut vertex would belong to $R = R(T)$.

For each vertex $v$ in $T$ we hold a variable $size(v)$.
Also, we initialize the set $CV_\ell$ to $\emptyset$.
The procedure visits the vertices of $T$ in a post-order manner, so that a
vertex $v$ is visited only after all its children have been visited.
For each visited vertex $v$, the procedure 
assigns $size(v) = 1 + \sum_{u \in Ch(v)} size(u)$ if $v \in R$,
and $size(v) = \sum_{u \in Ch(v)} size(u)$ otherwise,
 where $Ch(v)$ denotes the set of children of $v$ in (the current tree) $T$. 
(If $v$ is a leaf, then $Ch(v) = NULL$, and so $size(v) = 1$ if $v \in R$, and $size(v) = 0$
otherwise.)
Also, if $size(v)
> \ell$, the procedure designates $v$ as a cut vertex by inserting it to $CV_\ell$,
and then removes the subtree $T_v$ of $T$ rooted at $v$ from 
$T$. (See Figure \ref{f:decom} for an illustration.)

\begin{figure*}[htp]
\begin{center}
\begin{minipage}{\textwidth} 
\begin{center}
\setlength{\epsfxsize}{4.6in} \epsfbox{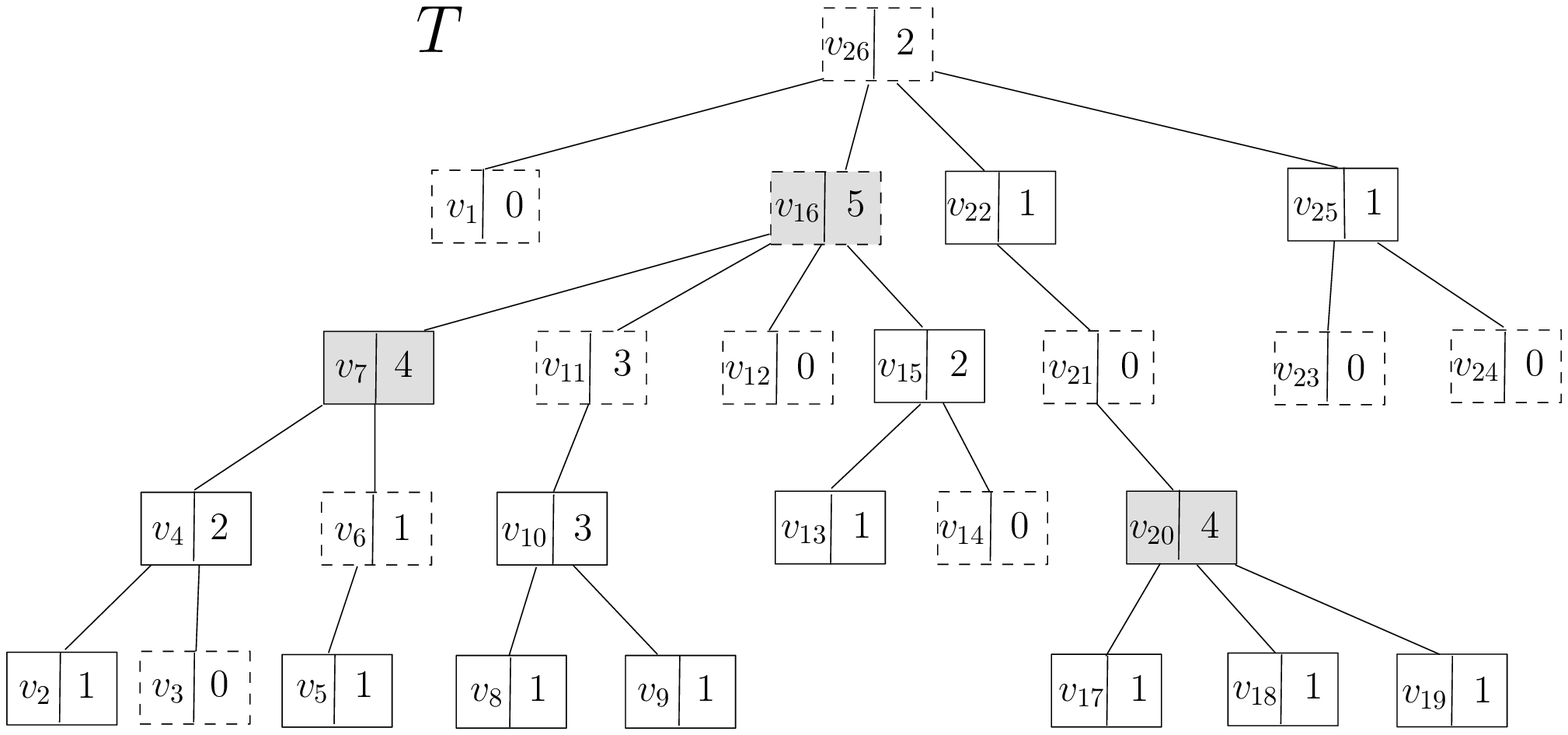}
\end{center}
\end{minipage}
\caption[]{ \label{f:decom} \sf \footnotesize An illustration of a rooted Steiner tree $(T,rt)$ over the vertices $v_1,v_2,\ldots,v_{26}$,
where $rt = v_{26}$, and for each $i \in [26]$, $v_i$ is the $i$th visited vertex in a (left-to-right) post-order traversal.
Each vertex $v_i$ in the tree is represented as a two-cell rectangle, with the left cell holding its name $v_i$,
and the right one holding the value of $size(v_i)$.
The 15 required vertices of the tree are depicted within solid lines, whereas the 11 Steiner vertices are depicted
within dashed lines. The three vertices whose bounding rectangles are filled, i.e., $v_7,v_{16}$ and $v_{20}$, comprise the set 
$CV_3$ of cut-vertices that is returned as output of the call $Decomp((T,rt),\ell=3)$.}
\end{center}
\end{figure*}

First, notice that the running time of the procedure $Decomp$ is linear in the number of vertices in $T$. 
In particular, if $T$ is pruned, then the running time of this procedure is $O(n)$.

We proceed by making the following observation.
\begin{observation} \label{ob:si}
At
the end of the execution of the procedure $Decomp$, for any subtree $\tau \in T \setminus CV_\ell$ and any vertex $w \in \tau$,
 $size(w)$ holds the required-size of the subtree $\tau_w$ of $\tau$ rooted at $w$, i.e., 
$size(w) = |R(\tau_w)|$. 
\end{observation}

Next, we obtain upper bounds on the maximum required-size of a subtree in $T \setminus CV_\ell$
and the size of the set $CV_\ell$ of cut vertices that is returned  by the procedure $Decomp$. 
\begin{lemma} \label{cen}
(1) The required-size $|R(\tau)|$ of any subtree $\tau \in T \setminus CV_\ell$ is at most $\ell$.
(2) $|CV_\ell| \le \lfloor \frac{n}{\ell+1} \rfloor$. 
\end{lemma}
\begin{proof}
(1) Consider an arbitrary subtree $\tau \in T \setminus CV_\ell$, and let $x$ be the root vertex of $\tau$. By the description of the procedure and Observation \ref{ob:si}, we have $|R(\tau)| = |R(\tau_x)| 
= size(x) \le \ell$, as otherwise $x$ would have been designated as a cut vertex.
\\(2) Immediately after a cut vertex $v$ is inserted into $CV_\ell$, the procedure removes the subtree $T_v$
of $T$ rooted at $v$ from $T$, and so the required-size
of the tree $T$ is being decreased by $|R(T_v)|$ units. Define $\chi(v) = 1$ if $v \in R$,
and $\chi(v)= 0$ otherwise. By
the description of the procedure and Observation \ref{ob:si}, just before the removal of $T_v$ from $T$ we have 
$$|R(T_v)| ~=~ \chi(v) + \sum_{u \in Ch(v)}|R(T_u)| ~=~ \chi(v) + \sum_{u \in Ch(v)}size(u) 
~=~ size(v) ~>~ \ell,$$
implying that  the required-size of $T$ is being decreased by at least $\ell+1$
units. Hence, after $i$ cut vertices have been inserted into $CV_\ell$,
the required size of $T$ is at most $n - i (\ell+1)$.
Also, from the moment the required-size of $T$ becomes at most $\ell$,
the set $CV_\ell$ remains intact, and we are done. \QED
\end{proof}
{\bf Remark:} The tradeoff $\ell$ versus $\lfloor \frac{n}{\ell+1} \rfloor$ between the 
required-size of a subtree in $T \setminus CV_\ell$ and the size of the set $CV_\ell$ 
of cut-vertices is tight, and is
realized when $T$ is the unweighted path graph $P_n$. 


\subsection{Sparse 1-Spanners for Tree Metrics with Bounded Diameter}
\label{s:sol}
In this section we present an optimal 
time construction of sparse 1-spanners for Steiner tree metrics with bounded diameter.
Our spanners achieve a tight tradeoff between the diameter and number of edges.

Let $(T,rt)$ be a Steiner rooted tree. Notice that $T$ can be transformed in linear time into
a pruned $T$-monotone preserving tree $(T',rt')$ by invoking the procedure $Prune$ described in Section \ref{s:prune} on $T$.
Also, any 1-spanner for $T'$ provides a 1-spanner for the original tree $T$ with the same diameter.
We may henceforth assume that the original tree $T$ is pruned, i.e., $T = T'$. We also assume that for each vertex $v$ in $T$, it can be decided in constant time whether it is black or white, i.e., whether $v \in R(T)$ or $v \in S(T)$.

Next, we describe an algorithm $Tree1Spanner = Tree1Spanner((T,rt),n,k)$ that accepts as input a pruned tree $(T,rt)$, an integer $n \ge 0$ that designates the required-size of $T$, and an integer
$k \ge 2$, and returns
as output a 1-spanner for $T$.

If $0 \le n \le k$, return the edge set $E(T)$ of $T$.

If $n = k+1$, check whether $rt$ has exactly two children.
  If this is the case, return $E(T) \cup \{(c_1,c_2)\}$, where $c_1$ and $c_2$ designate the two children of $rt$.
  Otherwise, return $E(T)$.
  
We henceforth assume that $n \ge k+2$. The execution of the algorithm splits into six steps.

At the first step, set $\ell = \alpha'_{k-2}(n)$, and compute the set $CV_\ell$ of cut vertices
of $T$ by 
making the call $Decomp((T,rt),\ell)$.

At the second step, compute the edge set $E'$ that connect the cut vertices.
\\If $k=2$, set $E' = \emptyset$.
If $k=3$, set $E'$ as the edge set of the complete graph over $CV_\ell$.
\\For $k \ge 4$, proceed in the following way. First, compute a copy $\tau$ of $T$. Second, go over all the vertices of $\tau$ and 
color the vertices of 
  $CV_\ell$ in black, and the remaining vertices in white. 
  (Thus $R(\tau) = CV_\ell$, and $S(\tau) = V(T) \setminus CV_\ell$.)
  Third, compute the pruning $\tau'$ of $\tau$ by making the call $Prune((\tau,rt))$.
  Fourth, set $E'$ as the edge set returned by the recursive call $Tree1Spanner((\tau',rt(\tau')),|CV_\ell|,k-2)$. 
 
 At the third step, compute the subtrees $T_1,\ldots,T_g$ 
 in $T \setminus CV_\ell$.
 
 At the fourth step, compute the edge set $E''$ that connects the cut vertices of $CV_\ell$
 with the corresponding
 subtrees.
Specifically, the set of all cut vertices $u \in CV_\ell$ that are connected by an edge of $T$
to some vertex of $T_i$ is called the \emph{border} of $T_i$, for each $i \in [g]$.
The vertex $u$ is called a \emph{border vertex} of $T_i$. 
Compute the edge set $E''=\{(u,v) ~\vert~ v \in R(T) \setminus CV_\ell, u \in CV_\ell, u \mbox{ is a border vertex
of the subtree to which } v \mbox{ belongs}\}.$
\\(See Figure \ref{border} for an illustration.)

\begin{figure*}[htp]
\begin{center}
\begin{minipage}{\textwidth} 
\begin{center}
\setlength{\epsfxsize}{3.6in} \epsfbox{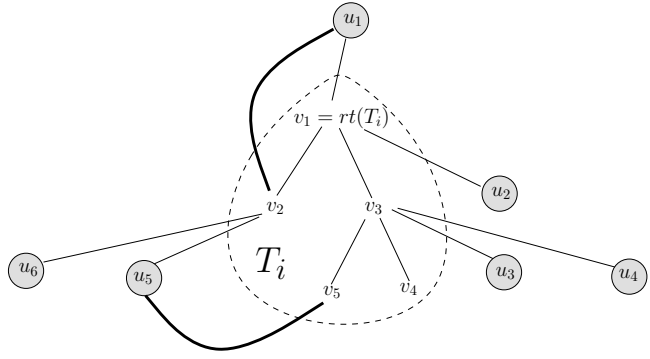}
\end{center}
\end{minipage}
\caption[]{ \label{border} \sf \footnotesize 
An illustration of a subtree $T_i \in T \setminus CV_\ell$ that contains the five vertices $rt = v_1,v_2,\ldots,v_5$,
with $v_3$ being the only Steiner vertex in $T_i$. The border of $T_i$
is comprised of the six vertices $u_1,u_2,\ldots,u_6$ that are depicted within filled circles.
Each required vertex $v_j$ in $T_i$, $j \in [5], j \ne 3$, is incident on the single upstream edge $(u_1,v_j)$,
and on the five downstream edges $(u_2,v_j),(u_3,v_j),(u_4,v_j),(u_5,v_j)$ and $(u_6,v_j)$. The upstream edge $(u_1,v_2)$ and the downstream edge $(u_5,v_5)$ are depicted by 
bold lines.}
\end{center}
\end{figure*}

At the fifth step we would like to proceed recursively for each of the subtrees $T_1,\ldots,T_g$.
To this end, first
compute the pruning $T'_i$ of the subtree $T_i$, for each $i \in [g]$, by making the call $Prune((T_i,rt(T_i))$. 
Then, set $E_i$ to be the edge set that is returned by the recursive call $Tree1Spanner((T'_i,rt(T'_i)),|R_i|,k)$,
for each $i \in [g]$, 
where $R_i = R(T_i)$.

Finally, at the sixth step, return the edge set 
$E = E' \cup E'' \cup \bigcup_{i=1}^g E_i$.

\ignore{
\begin{algorithm}[H]
\caption{$Tree1Spanner((T,rt),k)$:}
\label{tree1span}
\begin{algorithmic}
  \REQUIRE $T$ is pruned, $k \ge 2$. 
  \STATE Let $n = |R(T)|$.
  \STATE If $1 \le n \le k$, return the edge set $E(T)$ of $T$.  
  \STATE If $n = k+1$, check whether $rt$ has exactly two children.
  If this is the case, return $E(T) \cup \{c_1,c_2\}$, where $c_1$ and $c_2$ designate these children.
  Otherwise, return $E(T)$.
  \STATE Assume that $n \ge k+2$.
  \STATE {\bf Step 1:} Let $\ell = \alpha'_{k-2}(n)$. Compute the set $CV_\ell$ of cut vertices by calling
  the procedure $~~~~Decomp((T,rt),\ell)$.
  \STATE {\bf Step 2:} 
  Compute a copy $\tau$ of $T$. Go over all the vertices of $\tau$ and designate the vertices of 
  $CV_\ell$ \STATE ~~~~~as the required vertices of $\tau$, and the other vertices as the Steiner vertices. 
  \{Thus $R(\tau) = CV_\ell$,
  \STATE ~~~~and $S(T) = V(T) \setminus CV_\ell$.\}
  Compute the pruning $\tau'$ of $\tau$ by making the call $Prune((\tau,rt))$.
  \STATE {\bf Step 3:} If $k=2$, set $E' = \emptyset$.
  If $k=3$, set $E'$ as the edge set of the complete graph over $CV_\ell$.
  \STATE ~~~~For $k \ge 4$, let $E'$ be edge set returned by the recursive call $Tree1Spanner((\tau',rt),k-2)$. 
  \STATE {\bf Step 4:} Compute the subtrees $T_1,\ldots,T_g$ obtained by removing the cut vertices of $CV_\ell$ from $T$.
  \STATE ~~~For each $i \in [g]$, compute the pruning $T'_i$ of $T_i$ by making the call $Prune((T_i,rt(T_i))$. \{Thus 
  \STATE ~~~$T'_i$ is pruned, with $R(T'_i) = R(T) \cap V(T'_i), S(T'_i) = S(T) \cap V(T'_i).$\}
  ~For each $i \in [g]$, let $E_i$ \STATE~~~be the edge set that is returned by the recursive call $Tree1Spanner((T'_i,rt(T'_i)),k)$.
  \STATE {\bf Step 5:} Compute the edge set $E'' = \{(u,v) ~\vert~ v \in R(T) \setminus CV_\ell, u \in CV_\ell, u$  
  is a border vertex of 
  \STATE ~~~~~~~~~~~~~~~~~~~~~~~~~~~~~~~~~~~~~~~~~~~~~~~~~~~~~~~~~~~~~the subtree in $T \setminus CV_\ell$ that has $v$ as a vertex\}.
  \STATE {\bf Step 6:} Return the edge set $E = E' \cup E'' \cup \bigcup_{i=1}^g E_i$.
\end{algorithmic}
\end{algorithm}
}

The following theorem summarizes the properties of Algorithm $Tree1Spanner((T,rt),n,k)$.

\begin{theorem} \label{kequalk}
Let $k \ge 2$ and $n \ge 0$ be two arbitrary integers, and let $(T,rt)$ be a pruned tree with required-size $n$. 
Algorithm $Tree1Spanner((T,rt),n,k)$ computes in time $O(n \alpha_k(n))$ 
a 1-spanner $G_{T} = (V(T),E)$ for $T$, having diameter at most $k$ and $O(n \alpha_k (n))$ edges. 
\end{theorem}
{\bf Remarks:}
(1) If we set $\ell = \alpha_{k-2}(n)$ instead of $\ell = \alpha'_{k-2}(n)$ at the first step
of the algorithm, then both
the running time of the algorithm and the number of edges in the resulting spanner $G_T$ would increase by a factor of $k$,
i.e., from $O(n \alpha_k(n))$ to $O(n k \alpha_k(n))$.
~(2) In Section \ref{s:value} we saw that $\alpha_{2\alpha(n)+2}(n) \le 4$.
Hence, we can compute in $O(n)$ time a 1-spanner for $T$ having diameter
at most $2\alpha(n)+2$ and $O(n)$ edges. 

In what follows we prove Theorem \ref{kequalk}.

The next lemma bounds the size of the edge
set $E''$ that is computed at the fourth 
step of the algorithm and the time needed to compute it. This lemma was essentially proved in \cite{BTS94,NS07}.
\begin{lemma} \label{bridges}
The edge set $E''$ contains at most $2n$ edges. Also, it can be computed in $O(n)$ time.
\end{lemma}
\begin{proof}
Every edge of $E''$ is incident on exactly one cut vertex.
Consider such an edge $(u,v) \in E''$, where $u \in CV_\ell$ and $v \in R(T) \setminus CV_\ell$.
Then $v$ belongs to some subtree $T_i$ in $T \setminus CV_\ell$.
We say that the edge $(u,v)$ is \emph{upstream} if $u$ is the parent of the root $rt(T_i)$
of the subtree to which $v$ belongs. Otherwise, the edge $(u,v)$ is called \emph{downstream}.
(See Figure \ref{border} for an illustration.)

By definition, each vertex $v \in R \setminus CV_\ell$ is incident on at most one upstream edge. 
Hence, there are at most $|R \setminus CV_\ell| \le |R| = n$
upstream edges in total.
\\The downstream edges are counted per cut vertex. 
Each cut vertex $u \in CV_\ell \setminus \{rt\}$ has one parent in $T$, denoted $\pi_T(u)$.
If $\pi_T(u) \in CV_\ell$, then no downstream edge is incident on $u$.
Otherwise, $\pi_T(u)$ belongs to some subtree $T_i \in T\setminus CV_\ell$. Each downstream edge that is
incident on $u$ belongs to a distinct required vertex in $T_i$. Hence, the first assertion of Lemma \ref{cen}
implies that $u$ is incident on at most $\ell$ downstream edges. By the second assertion of Lemma \ref{cen},
$|CV_\ell| \le \lfloor \frac{n}{\ell+1}\rfloor$. Summing over all vertices in $CV_\ell \setminus \{rt\}$,
we get a total of at most $\lfloor \frac{n}{\ell+1}\rfloor \ell \le n$ downstream edges.
Hence, there are overall at most $2n$ edges in $E''$.

To verify that $E''$ can indeed be constructed within $O(n)$ time, we refer to Exercise 12.4 in \cite{NS07}. 
\QED
\end{proof}

Next, we prove Theorem \ref{kequalk} in the particular case of $k=2$.
\begin{lemma} \label{kequal2}
Let $(T,rt)$ be a pruned tree with required-size $n \ge 0$. 
Algorithm $Tree1Spanner((T,rt),n,2)$ computes in $O(n \alpha_2(n))$ time  
a 1-spanner $G_T = (V(T),E)$ for $T$, having diameter at most 2 and at most $n \alpha_2 (n)$ edges. 
\end{lemma}
\begin{proof}
We denote by $F_{2}(n)$ the maximum number
of edges in the graph computed by Algorithm $Tree1Spanner((T,rt),n,2)$, where
$T$ ranges over all pruned trees having required-size $n$.
We next prove by induction on $n$ that $F_2(n) \le n \alpha_2(n)$.  
Let $T$ be a pruned tree with required-size $n$ for which the edge set $E$ that is computed
by Algorithm $Tree1Spanner((T,rt),n,2)$ has $F_2(n)$ edges.
\\The case $n=0$ is trivial. We henceforth assume that $n \ge 1$.
\\By Lemma \ref{compa}, any non-empty pruned tree is compact. Hence, $|V(T)| \le 2|R(T)|-1 = 2n-1$, and so $|E(T)| = |V(T)|-1 \le 2n-2$.
\\If $n \le 2$, then $F_2(n) ~=~ |E| ~=~ |E(T)| ~\le~ 2n-2$.
If $n=1$, then $\alpha_2(n) = 0$, yielding 
$F_2(n)  \le 2n-2 = 0 = n \alpha_{2}(n)$. 
If $n=2$, then 
$\alpha_{2}(n) = 1$, yielding $F_2(n) \le 2n-2 = 2 = n \alpha_2(n)$.
\\Suppose next that $n=3$. In this case the edge set $E$ returned by the algorithm 
contains at most one more edge in addition to the edge set $E(T)$ of the input tree $T$.
Hence, $F_2(n) ~=~ |E| ~\le~ |E(T)| + 1 ~\le~ 2n-1 ~=~ 5$. Notice that $\alpha_2(3) = 2$,
yielding
$F_2(n) \le 5 \le n \alpha_{2}(n)$.
\\\emph{Induction Step:} We assume the correctness of the statement for all smaller
values of $n$, $n\ge 4$, and prove it for $n$.
Note that $\ell = \alpha'_0(n) = \alpha_0(n) = \lceil n/2 \rceil$.
By the second assertion of Lemma \ref{cen}, $|CV_\ell| \le \left\lfloor \frac{n}{\ell+1} \right\rfloor \le \left\lfloor \frac{n}{n/2 +1}\right\rfloor =1$,
implying that $CV_\ell$ consists of a single vertex, denoted $w$.
\\Observe that for $k=2$, the edge set $E'$ is empty.
\\Since $CV_\ell$ consists of a single vertex $w$, the edge set $E''$ that is computed at the fourth step of the algorithm is comprised of all edges that
connect $w$ to the required vertices in $R(T) \setminus \{w\}$. Hence, $|E''| = |R(T) \setminus \{w\}| \le |R(T)| = n$.
\\Let $i$ be an index in  $[g]$, and consider the edge set $E_i$ that is computed at the fifth
step of the algorithm. 
We have $|E_{i}| \le F_{2}(|R_i|)$. 
By the first assertion of Lemma \ref{cen}, the required-size of each subtree in $T \setminus CV_\ell = T \setminus \{w\}$
is at most $\ell$, and so 
$|R_i| \le \ell  = \lceil n/2 \rceil < n$.
Since the function $\alpha_{2}$ is monotone non-decreasing, the induction hypothesis implies that $|E_i|
\le |R_i| \alpha_{2}(\ell)$. Since $n \ge 4$, we have $\alpha_{2}(n) = 1+\alpha_{2}(\alpha_{0}(n))
= 1 + \alpha_{2}(\ell)$. Also, notice that $\sum_{i=1}^g |R_i| \le |R| = n$. 
It follows that $$\sum_{i=1}^g |E_i| ~\le~ \sum_{i=1}^g |R_i| \alpha_{2}(\ell)
~=~ \sum_{i=1}^g |R_i| (\alpha_{2}(n) -1) ~\le~ n(\alpha_{2}(n) -1).$$
Altogether,
\begin{eqnarray*} F_2(n) ~=~ |E| ~=~ |E'| + |E''| + \sum_{i=1}^g |E_i| ~\le~
0 + n + n(\alpha_{2}(n) -1) ~=~ n \alpha_{2}(n). 
\end{eqnarray*}

Next, we prove that $G_T$ is a 1-spanner for $T$ with diameter at most 2.
The proof is, again, by induction on $n$.
The case $n \le 3$ follows from Lemma \ref{hop} and Corollary 
\ref{corhop}.
\\\emph{Induction Step:} We assume the correctness of the statement for all smaller
values of $n$, $n \ge 4$, and prove it for $n$. 
We show that for an arbitrary pair $u,v$ of required vertices, there is a $T$-monotone
path in $G_T$ that consists of at most two edges.
Consider the single vertex $w$ in $CV_\ell$. If either $u$ or $v$ is equal to $w$, then $u$ and $v$ are connected by an edge of $E''$,
and so this edge forms a $T$-monotone path between $u$ and $v$.
If $u$ and $v$ are in different subtrees of $T \setminus \{w\}$, then both edges $(u,w)$ and $(w,v)$ belong to $G_T$,
and so $u$ and $v$ are connected by the 
path $P = (u,w,v)$ in $G_T$. 
Notice that the unique path $P_T(u,v)$ between $u$ and $v$ must traverse
$w$, implying that $P$ is $T$-monotone. Finally, if $u$ and $v$ belong to the same
subtree $T_i$ in $T \setminus \{w\}$, then by the induction hypothesis they are connected by a $T_i$-monotone
path $P_i$ that consists of at most two edges. However, $P_i$ is also a path in $G_T$, and it is $T$-monotone.

Denote by $C_2(n)$ the worst-case running time
of Algorithm 
$Tree1Spanner((T,rt),n,2)$, where $T$ ranges over all pruned trees $T$ with required-size $n$.
We next show that $C_2(n) =  O(n \alpha_2(n))$. 
\\Clearly, if $n \le 3$, then $C_2(n) = O(1)$. We may henceforth assume that
$n \ge 4$.
\\Computing the set $CV_\ell$ of cut vertices at the first step of the algorithm
takes $O(n)$ time. 
Also, $E' = \emptyset$, and so the second step of the algorithm requires only $O(1)$ time.
It takes $O(n)$ time to compute the subtrees $T_1,\ldots,T_g$ and the corresponding pruned subtrees $T'_1,\ldots,T'_g$ at the third and fifth steps
of the algorithm, respectively. Recall that $CV_\ell$ consists of a single vertex $w$, and so one can compute 
the edge set $E'' = \{(w,v) ~\vert~ v \in R(T) \setminus \{w\}\}$ directly within $O(n)$ time
as well.
Finally, the time needed to compute the edge
sets $E_1,E_2,\ldots,E_g$ at the fifth step of the algorithm is at most $\sum_{i=1}^g C_2(|R_i|)$.
We obtain the recurrence $C_2(n) = O(n) + \sum_{i=1}^g C_2(|R_i|)$, where $|R_i| \le \ell = \lceil n/2 \rceil$,
for each index $i \in [g]$, and $\sum_{i=1}^g |R_i| \le n$.
Hence, as in the above argument for bounding $F_2(n)$, it can be shown that
$C_2(n) = O(n \alpha_2(n))$. 
\QED
\end{proof}

Next, we prove Theorem \ref{kequalk} in the particular case of $k=3$.
\begin{lemma} \label{kequal3}
Let $(T,rt)$ be a pruned tree with required-size $n \ge 0$. 
Algorithm $Tree1Spanner((T,rt),n,3)$ computes in $O(n \alpha_3(n))$ time 
a 1-spanner $G_T = (V(T),E)$ for $T$, having diameter at most 3 and $\frac{5}{2} n \alpha_3 (n)+2$ edges. 
\end{lemma}
\begin{proof}
We denote by $F_{3}(n)$ the maximum number
of edges in the graph computed by Algorithm $Tree1Spanner((T,rt),n,3)$, where
$T$ ranges over all pruned trees having required-size $n$.
We next prove by induction on $n$ that $F_3(n)$ is no greater than $\max\{2,\frac{5}{2} n \alpha_3(n)\}$, 
which provides the required result.
Let $T$ be a pruned tree with required-size $n$ for which the edge set $E$ that is computed
by Algorithm $Tree1Spanner((T,rt),n,3)$ has $F_3(n)$ edges.
\\The case $n=0$ is trivial. We henceforth assume that $n \ge 1$.
\\Notice that $\max\{2,\frac{5}{2} n \alpha_3 (n)\} = \frac{5}{2} n \alpha_3 (n)$, for all $n \ge 3$.
Also, by Lemma \ref{compa}, every non-empty pruned tree is compact. Hence, $|V(T)| \le 2|R(T)|-1 = 2n-1$, and so $|E(T)| = |V(T)|-1 \le 2n-2$.
\\If $n \le 3$, then $F_3(n) ~=~ |E| ~=~ |E(T)| ~\le~ 2n-2$.
For $n \le 2$, we have $F_3(n) \le 2n-2 \le 2$.
For $n =3$, we have $\alpha_3(3) = 1$, and so $F_3(n) \le 2n-2 = 4 \le \frac{5}{2} n \alpha_3(n)$. 
In both cases, we have $F_3(n)  \le \max\{2,\frac{5}{2} n \alpha_3 (n)\}$.
\\Suppose next that $n=4$. In this case the edge set $E$ returned by the algorithm 
contains at most one more edge in addition to the edge set $E(T)$ of the input tree $T$.
Hence, $F_3(n) ~=~ |E| ~\le~ |E(T)| + 1 ~\le~ 2n-1 ~=~ 7$. Notice that $\alpha_3(4) = 1$,
and so
$F_3(n) \le 7 \le \frac{5}{2} n \alpha_{3}(n) = \max\{2,\frac{5}{2} n \alpha_3 (n)\}$.
\\\emph{Induction Step:} We assume the correctness of the statement for all smaller
values of $n$, $n \ge 5$, and prove it for $n$. 
Observe that $\ell = \alpha'_1(n) = \alpha_1(n) = \lceil \sqrt{n} \rceil$.
By the second assertion of Lemma \ref{cen}, $|CV_\ell| \le \left\lfloor \frac{n}{\ell+1} \right\rfloor \le \sqrt{n}$.
\\Observe that for $k=3$, the edge set $E'$ consists of all ${|CV_\ell| \choose 2}$ edges of the complete graph
over $CV_\ell$, and so $|E'| = {|CV_\ell| \choose 2} \le  {\sqrt{n} \choose 2} \le \frac{n}2$.
\\By Lemma \ref{bridges}, 
the number of edges in the edge set $E''$ that is computed at the fourth step of
the algorithm is less than or equal to $2n$. 
\\Let $I^-_1$, $I_2$, and $I^+_3$ be the sets of all indices $i \in [g]$ for which $|R_i| \le 1$, $|R_i| = 2$,
and $|R_i| \ge 3$, respectively. Clearly, $I^-_1 \cup I_2 \cup I^+_3 = [g]$. 
Observe that $$n ~=~ |R| ~\ge~ \sum_{i=1}^g |R_i| ~=~ \sum_{i \in I^-_1} |R_i| + \sum_{i \in I_2} |R_i|
+ \sum_{i \in I^+_3} |R_i| ~\ge~ 2|I_2| + \sum_{i \in I^+_3} |R_i|,$$ implying that
$\sum_{i \in I^+_3} |R_i| \le n- 2|I_2|$. 
Let $i$ be an index in $[g]$, and consider the edge set $E_i$ that is computed at the fifth step
of the algorithm. 
We have $|E_{i}| \le F_{3}(|R_i|)$. Observe that if $i \in I^-_1$,
then   $|E_i| = E(T'_i) = 0$, and if $i \in I_2$, then  $|E_i| = E(T'_i) \le 2$.
Suppose next that $i \in I^+_3$. 
By the first assertion of Lemma \ref{cen}, the required-size of each subtree in $T \setminus CV_\ell$
is at most $\ell$, and so
$3 \le |R_i| \le \ell =  \lceil \sqrt{n} \rceil < n$. Since the function $\alpha_{3}$ is monotone non-decreasing, the induction hypothesis implies that $|E_i|
\le \max\{2,\frac{5}{2}|R_i| \alpha_{3}(\ell)\} = \frac{5}{2}|R_i| \alpha_{3}(\ell)$. 
Since $n \ge 5$, we have $2 \le \alpha_{3}(n) = 1+\alpha_{3}(\alpha_{1}(n)) = 1 + \alpha_{3}(\ell)$. 
It follows that \begin{eqnarray*} \sum_{i=1}^g |E_i| &=& \sum_{i \in I^-_1} |E_i| + \sum_{i \in I_2} |E_i|
+ \sum_{i \in I^+_3} |E_i| ~\le~ 2|I_2| +   \sum_{i \in I^+_3} |E_i|
 ~\le~ 2|I_2| + \sum_{i \in I^+_3} \frac{5}{2}|R_i| \alpha_{3}(\ell) \\ &=& 2|I_2| + \sum_{i \in I^+_3} \frac{5}{2}|R_i| (\alpha_{3}(n)-1) ~\le~ 2|I_2| + \frac{5}{2}(n-2|I_2|) (\alpha_{3}(n)-1) \\ &=& 2|I_2| + \frac{5}{2} n (\alpha_{3}(n)-1)
- 5|I_2|(\alpha_3(n)-1) ~\le~ \frac{5}{2} n (\alpha_{3}(n)-1).
\end{eqnarray*}
(The last inequality holds since $\alpha_3(n) \ge 2$.)
Altogether, \begin{eqnarray*} F_3(n) ~=~ |E| ~=~ |E'| + |E''| + \sum_{i=1}^g |E_i| ~\le~
\frac{n}{2} + 2n + \frac{5}{2}n(\alpha_{3}(n) -1) ~=~ \frac{5}{2}n \alpha_{3}(n) = \max\left\{2,\frac{5}{2} n \alpha_{3}(n)\right\}. 
\end{eqnarray*}

Next, we prove that $G_T$ is a 1-spanner for $T$ with diameter at most 3.
The proof is, again, by induction on $n$.
The case $n \le 4$ follows from Lemma \ref{hop} and Corollary \ref{corhop}.
\\\emph{Induction Step:} We assume the correctness of the statement for all smaller
values of $n$, $n \ge 5$, and prove it for $n$. 
Observe that for $k=3$, the edge set $E'$ is equal to the edge set of the complete graph over $CV_\ell$,
and so there is an edge in $G_T$ between any pair of vertices in $CV_\ell$.
\\ Next, we show that for an arbitrary pair $u,v$ of required vertices, there is a $T$-monotone
path in $G_T$ that consists of at most three edges. 
The analysis splits into five cases.
\\\emph{Case 1: $u,v \in CV_\ell$.} 
In this case there is an edge in $G_T$ between $u$ and $v$, which forms a $T$-monotone path.
\\\emph{Case 2: $u \in CV_\ell$ and $v \in R(T) \setminus CV_\ell$}. 
Let $w$ be the first vertex of $CV_\ell$ on the path in $T$ from $v$ to $u$.
Note that $w$ is a border vertex of the subtree in $T \setminus CV_\ell$ that has $v$ as a vertex,
and so the edge $(w,v)$ belongs to $E''$, and thus also to $G_T$.
If $u=w$, then $(v,u)$ is an edge in $G_T$, which forms a $T$-monotone path. Otherwise $u \ne w$.
Note that both $w$ and $u$ belong to $CV_\ell$, and so there is an edge in $G_T$ between $u$ and $w$.
Hence the two edges $(u,w)$ and $(w,v)$ form a $T$-monotone path $(u,w,v)$ between $u$ and $v$ that consists of
two edges.
\\\emph{Case 3: $v \in CV_\ell$ and $u \in R(T) \setminus CV_\ell$}. This case is symmetrical to case 2.
\\\emph{Case 4: $u \in T_i$, $v \in T_j$, for two distinct subtrees $T_i$ and $T_j$ in $T \setminus CV_\ell$.}
Let $w$ and $w'$ be the first and last vertices of $CV_\ell$ on the path in $T$ from $u$ to $v$,
respectively. Note that $w$ is a border vertex of $T_i$ and $w'$ is a border vertex of $T_j$,
and so both edges $(u,w)$ and $(w',v)$ belong to $E''$, and thus also to $G_T$. If $w = w'$, then
$(u,w,v)$ is a $T$-monotone path between $u$ and $v$ in $G_T$ that consists of two edges.
Otherwise, $w \ne w'$. Note that both $w$ and $w'$ belong to $CV_\ell$, and so they are connected by the edge $(w,w')$ in $G_T$.
Hence the three edges $(u,w)$, $(w,w')$, and $(w',v)$ form a $T$-monotone path $(u,w,w',v)$ between $u$ and $v$
that consists of three edges.
\\\emph{Case 5: $u,v \in T_i$, for some subtree $T_i$ in $T \setminus CV_\ell$.}
The first assertion of Lemma \ref{cen} implies that 
the required-size $|R_i| = |R(T_i)|$ of $T_i$ is at most
$\ell = \alpha'_1(n) = \lceil \sqrt{n} \rceil < n$. By Lemma \ref{description}, the tree $T'_i$ that is computed at the fifth step of the algorithm
satisfies $R(T'_i) = R(T_i)$.
Hence, by the induction hypothesis for $T'_i$, the $T'_i$-monotone diameter of the 
graph $G_{T'_i} = (V(T'_i),E_i)$ that is computed at the fifth step of the algorithm is at most
$3$. It follows that
$u$ and $v$ are connected in $G_T$ by a $T'_i$-monotone path that consists of at most three edges.
However, since $T'_i$ is $T_i$-monotone preserving, this path is also $T_i$-monotone, and thus also $T$-monotone.

Denote by $C_3(n)$ the worst-case running time
of Algorithm 
$Tree1Spanner((T,rt),n,3)$, where $T$ ranges over all pruned trees $T$ with required-size $n$.
We next show that $C_3(n) =  O(n \alpha_3(n))$. 
\\Clearly, if $n \le 4$, then $C_3(n) = O(1)$. We may henceforth assume that
$n \ge 5$.
\\Computing the set $CV_\ell$ of cut vertices at the first step of the algorithm
takes $O(n)$ time. 
Also, computing the edge set $E'$ of the complete
graph over $CV_\ell$ at the second step of the algorithm can be carried out in $O(|E'|) = O(n)$ time.
It takes $O(n)$ time to compute the subtrees $T_1,\ldots,T_g$ and the corresponding pruned subtrees $T'_1,\ldots,T'_g$ at the third and fifth steps
of the algorithm, respectively. By Lemma \ref{bridges},  
computing the edge set $E''$ at the fourth step of the algorithm takes  $O(n)$ time
as well.
Finally, the time needed to compute the edge
sets $E_1,E_2,\ldots,E_g$ at the fifth step of the algorithm is at most $\sum_{i=1}^g C_3(|R_i|)$.
We obtain the recurrence $C_3(n) = O(n) + \sum_{i=1}^g C_3(|R_i|)$, where $|R_i| \le \ell = \lceil \sqrt{n} \rceil$,
for each index $i \in [g]$, and $\sum_{i=1}^g |R_i| \le n$.
Hence, as in the above argument for bounding $F_3(n)$, it can be shown that
$C_3(n) = O(n \alpha_3(n))$. 
\QED
\end{proof}

We turn to prove Theorem \ref{kequalk} for a general $k$, $k \ge 2$.

The following lemma establishes an upper bound on the number of edges in the spanner $G_T$.
\begin{lemma} \label{edgebound}
Let $k \ge 2$ and $n \ge 0$ be two arbitrary integers,
and denote by $F_{k}(n)$ the maximum number
of edges in the graph computed by Algorithm $Tree1Spanner((T,rt),n,k)$, where
$T$ ranges over all pruned trees having required-size $n$.
Then $F_{k}(n) \le 2 n\alpha'_{k}(n)$ if $k$ is even, and $F_{k}(n) \le 3 n \alpha'_{k}(n) + 2$ otherwise.
\end{lemma}
{\bf Remark:} By Lemma \ref{finallyreal}, for all $k \ge 2$ and $n \ge 0$, $\alpha'_k(n) \le 2\alpha_k(n) + 4$. Hence,
$F_k(n)= O(n \alpha_k(n))$.
\begin{proof}
We first give the proof for even values of $k$. 
The proof is by double induction on $k$ and $n$.  
Let $T$ be a pruned tree with required-size $n$ for which the edge set
$E$ that is computed by algorithm $Tree1Spanner((T,rt),n,k)$ has $F_{k}(n)$ edges.
\\The case $n=0$ is trivial. Also,
the case $k=2$ follows from Lemma \ref{kequal2}. \\We henceforth assume that $n \ge 1$ and $k \ge 4$.
By Lemma \ref{compa}, every non-empty pruned tree is compact. Hence, $|V(T)| \le 2|R(T)|-1 = 2n-1$, and so $|E(T)| = |V(T)|-1 \le 2n-2$.
\\If $n=1$, then $F_{k}(n) ~=~ |E| ~=~ |E(T)| ~\le~ 2n-2 ~=~ 0$.
Also, $\alpha'_{k}(1) = 0$.
Hence, $F_{k}(n) = 0 = 2n \alpha'_{k}(n)$. 
\\Suppose next that $2 \le n \le k+1$. In this case the edge set $E$ returned by the algorithm 
contains at most one more edge in addition to the edge set $E(T)$ of the input tree $T$,
and so $F_{k}(n) ~=~|E| ~\le~ |E(T)|+1 ~\le~ 2n-1$.
Hence, $F_{k}(n) \le 2n-1 \le  2 n \alpha'_{k}(n)$,
as $\alpha'_{k}(n) \ge \alpha_{k}(n) \ge 1$, for all $n \ge 2$ and $k \ge 4$.
\\\emph{Induction Step:} We assume that for an arbitrary pair $(k,n)$, $n \ge k + 2  \ge 6$,
the statement holds for all pairs $(k',n')$,
with either $k' < k$ or both $k' = k$ and $n' < n$, and prove it for the pair $(k,n)$.
\\The number of edges in the set $E'$ that is computed at the second step of the algorithm is less
than or equal to $F_{k-2}(|CV_\ell|)$. Since $n \ge 6$, it holds that $\ell = \alpha'_{k-2}(n) \ge \alpha_{k-2}(n) \ge 1$.
By the second assertion of Lemma \ref{cen}, we have $|CV_{\ell}| \le \lfloor \frac{n}{\ell+1} \rfloor < n$. 
Since the function $\alpha'_{k-2}$ is monotone non-decreasing, $\alpha'_{k-2}(|CV_{\ell}|) \le \alpha'_{k-2}(n) = \ell$.
Therefore, by the induction hypothesis for the pair $(k-2,|CV_\ell|)$, 
$$|E'| ~\le~ F_{k-2}(|CV_\ell|) ~\le~ 2 |CV_{\ell}|\alpha'_{k-2}(|CV_{\ell}|) ~\le~
2\left\lfloor \frac{n}{\ell+1} \right\rfloor \ell ~\le~ 2n.$$
\\By Lemma \ref{bridges}, the number of edges in the set $E''$ that is computed at the fourth step of
the algorithm is less than or equal to $2n$. 
\\Let $i$ be an index in $[g] $, and consider the edge set $E_i$ that is computed at the fifth step
of the algorithm. 
We have $|E_{i}| \le F_{k}(|R_i|)$. 
The first assertion of Lemma \ref{cen} and the second assertion of Lemma \ref{base'real} imply that 
$|R_i| \le \ell = \alpha'_{k-2}(n) < n$. Since the function $\alpha'_{k}$ is monotone non-decreasing, 
the induction hypothesis for the pair $(k,|R_i|)$ implies that
$|E_i| \le 2 |R_i| \alpha'_{k}(\ell)$. Since $n \ge k+2$, we have $\alpha'_{k}(n) = 2+\alpha'_{k}(\alpha'_{k-2}(n))
= 2 + \alpha'_{k}(\ell)$. Also, notice that $\sum_{i=1}^g |R_i| \le |R| = n$. 
It follows that $$\sum_{i=1}^g |E_i| ~\le~ \sum_{i=1}^g 2|R_i| \alpha'_{k}(\ell)
~=~ \sum_{i=1}^g 2|R_i| (\alpha'_{k}(n) -2) \le 2n(\alpha'_{k}(n) -2).$$
Altogether,
\begin{eqnarray*}F_{k}(n) &=& |E| ~=~ |E'| + |E''| + \sum_{i=1}^g |E_i| ~\le~
2n + 2n + 2n(\alpha'_{k}(n) -2) ~=~ 2n \alpha'_{k}(n). 
\end{eqnarray*}

We next prove the lemma for odd values of $k$. 
The proof is, again, by double induction on $k$ and $n$. 
Let $T$ be a pruned tree with required-size $n$ for which the edge set
$E$ that is computed by algorithm $Tree1Spanner((T,rt),n,k)$ has $F_{k}(n)$ edges.
\\The case $n=0$ is trivial. Also,
the case $k=3$ follows from Lemma \ref{kequal3}. 
\\We henceforth assume that $n \ge 1$ and $k \ge 5$. 
By Lemma \ref{compa}, every non-empty pruned tree is compact. Hence, $|V(T)| \le 2|R(T)|-1 = 2n-1$, and so $|E(T)| = |V(T)|-1 \le 2n-2$.
\\If $n\le 2$, then $F_{k}(n) ~=~ |E| ~=~ |E(T)| ~\le~ 2n-2 ~\le~ 2 ~\le~ 3n \alpha'_{k}(n) + 2$.
\\Suppose next that $3 \le n \le k+1$. In this case the edge set returned by the algorithm 
consists of at most one more edge in addition to the edge set $E(T)$ of the input tree $T$,
and so $F_{k}(n) ~=~ |E| ~\le~ |E(T)|+1 ~\le~ 2n-1$.
Hence, $F_{k}(n) \le 2n-1 \le  3 n \alpha'_{k}(n) + 2$,
as $\alpha'_{k}(n) \ge \alpha_{k}(n) \ge 1$, for all $n \ge 3$ and $k \ge 5$.
\\\emph{Induction Step:} We assume that for an arbitrary pair $(k,n)$, $n \ge k+2 \ge 7$,
the statement holds for all pairs $(k',n')$,
with either $k' < k$ or both $k' = k$ and $n' < n$, and prove it for the pair $(k,n)$.
\\The number of edges
in the set $E'$ that is computed at the second step of the algorithm is less than or equal to
$F_{k-2}(|CV_\ell|)$. Since $n \ge 7$, it holds that $\ell = \alpha'_{k-2}(n) \ge \alpha_{k-2}(n) \ge 1$.
By the second assertion of Lemma \ref{cen}, we have $|CV_{\ell}| \le \lfloor \frac{n}{\ell+1} \rfloor < n$.  
Since the function $\alpha'_{k-2}$ is monotone non-decreasing, $\alpha'_{k-2}(|CV_{\ell}|) \le \alpha'_{k-2}(n) = \ell$.
Therefore, by the induction hypothesis for the pair $(k-2,|CV_\ell|)$,  
$$|E'| ~\le~ F_{k-2}(|CV_\ell|) ~\le~ 3|CV_{\ell}| \alpha'_{k-2}(|CV_{\ell}|) + 2
~\le~ 3\left\lfloor \frac{n}{\ell+1} \right\rfloor \ell + 2 ~\le~ 3n+2.$$
\\By Lemma \ref{bridges}, the number of edges in the set $E''$ that is computed at the fourth step of
the algorithm is less than or equal to $2n$.
\\Let $I^-_1$ (respectively, $I^+_2$) be the set of all indices $i$, such that $i \in [g]$ and $|R_i| \le 1$ (resp., $|R_i| \ge 2$).
Clearly, $I^-_1 \cup I^+_2 = [g]$.
Observe that $$n ~=~ |R| ~\ge~ \sum_{i=1}^g |R_i| ~=~ \sum_{i \in I^-_1} |R_i| + \sum_{i \in I^+_2} |R_i| ~\ge~ 2|I^+_2|.$$
Let $i$ be an index in $[g]$, and consider the edge set $E_i$ that is computed at the fifth step
of the algorithm.   
We have $|E_{i}| \le F_{k}(|R_i|)$. 
Observe that if $i \in I^-_1$, we have $|E_i| = |E(T'_i)| = 0$.
Suppose next that $i \in I^+_2$. 
The first assertion of Lemma \ref{cen} and the second assertion of Lemma \ref{base'real} imply that 
$|R_i| \le \ell = \alpha'_{k-2}(n) < n$.
Since the function $\alpha'_{k}$ is monotone non-decreasing, the induction hypothesis for the pair $(k,|R_i|)$ 
implies that $|E_i|
\le 3|R_i| \alpha'_{k}(\ell) + 2$. 
Since $n \ge k+2$, we have $\alpha'_{k}(n) = 2+\alpha'_{k}(\alpha'_{k-2}(n))
= 2 + \alpha'_{k}(\ell)$. 
It follows that
\begin{eqnarray*}\sum_{i=1}^g |E_i| ~=~ \sum_{i \in I^+_2} |E_i| &\le& \sum_{i \in I^+_2} \left(3|R_i| \alpha'_{k}(\ell) + 2\right)
~=~ \sum_{i \in I^+_2} 3|R_i| (\alpha'_{k}(n) -2) + 2|I^+_2| \\ &\le& 3n (\alpha'_{k}(n) -2) + n.
\end{eqnarray*}
Altogether,
\begin{eqnarray*}F_{k}(n) &=& |E| ~=~ |E'| + |E''| + \sum_{i=1}^g |E_i| ~\le~
(3n +2) + 2n + 3n (\alpha'_{k}(n) -2) + n ~=~ 3n \alpha'_{k}(n) +2. 
\end{eqnarray*}
\QED
\end{proof}

Next, we demonstrate that $G_T$ is a 1-spanner for $T$ with diameter at most $k$.
\begin{lemma} \label{thm2}
Let $k \ge 2$ and $n \ge 0$ be two arbitrary integers. For any pruned tree $(T,rt)$ with required-size $n$,
the $T$-monotone diameter $\Lambda(G_T)$ of 
the graph $G_T = (V(T),E)$ that is computed by Algorithm $Tree1Spanner((T,rt),n,k)$
is at most $k$.
\end{lemma}
\begin{proof}
The proof is by double induction on $k$ and $n$.
\\The cases $k=2$ and $k=3$ follow from Lemmas \ref{kequal2} and \ref{kequal3}, respectively.
\\We henceforth assume that $k \ge 4$.
\\For $0 \le n \le k+1$, the correctness of the statement
follows from Lemma \ref{hop} and Corollary \ref{corhop}.
\\\emph{Induction Step:} We assume that for an arbitrary pair $(k,n)$, $n \ge k+2 \ge 6$,
the statement holds for all pairs $(k',n')$,
with either $k' < k$ or both $k' = k$ and $n' < n$, and prove it for the pair $(k,n)$.
\\By Lemma \ref{description}, 
the tree $\tau'$ that is constructed at the second step of the algorithm satisfies $R(\tau') = R(\tau) = CV_\ell$.
By the induction hypothesis for the pair $(k-2,|CV_\ell|)$, the $\tau'$-monotone diameter of
the graph $G_{\tau'} = (V(\tau'),E')$ that is computed at the second step of the algorithm is at most $k-2$.
Since $\tau'$ is $\tau$-monotone-preserving
 and $\tau$ is a copy of $T$, it follows that 
 there is a $T$-monotone path in $G_T$ between any pair of vertices of $CV_\ell$ that consists of at most $k-2$ edges.
\\Next, we show that for an arbitrary pair $u,v$ of required vertices, there is a $T$-monotone
path in $G_T$ that consists of at most $k$ edges. The analysis splits into five cases.
\\\emph{Case 1: $u,v \in CV_\ell$.} 
In this case there is a $T$-monotone path in $G_T$ between $u$ and $v$
that consists of at most $k-2$ edges.
\\\emph{Case 2: $u \in CV_\ell$ and $v \in R(T) \setminus CV_\ell$}. 
Let $w$ be the first vertex of $CV_\ell$ on the path in $T$ from $v$ to $u$.
Note that $w$ is a border vertex of the subtree in $T \setminus CV_\ell$ that has $v$ as a vertex,
and so the edge $(w,v)$ belongs to $E''$, and thus also to $G_T$.
If $u=w$, then $(u,v)$ is an edge in $G_T$, which forms a $T$-monotone path. Otherwise $u \ne w$.
Since both $w$ and $u$ belong to $CV_\ell$, there is a $T$-monotone path in $G_T$ between $u$ and $w$ that consists of at most $k-2$
edges. Together with the edge $(w,v)$, we get a $T$-monotone path between $u$ and $v$ that consists of
at most $k-1$ edges.
\\\emph{Case 3: $v \in CV_\ell$ and $u \in R(T) \setminus CV_\ell$}. This case is symmetrical to case 2.
\\\emph{Case 4: $u \in T_i$, $v \in T_j$, for two distinct subtrees $T_i$ and $T_j$ in $T \setminus CV_\ell$.}
Let $w$ and $w'$ be the first and last vertices of $CV_\ell$ on the path in $T$ from $u$ to $v$,
respectively. Note that $w$ is a border vertex of $T_i$ and $w'$ is a border vertex of $T_j$,
and so both edges $(u,w)$ and $(w',v)$ belong to $E''$, and thus also to $G_T$. If $w = w'$, then
$(u,w,v)$ is a $T$-monotone path between $u$ and $v$ in $G_T$ that consists of two edges.
Otherwise, $w \ne w'$. Since both $w$ and $w'$ belong to $CV_\ell$, the graph $G_T$ contains a $T$-monotone path between $w$ and $w'$
that consists of at most $k-2$ edges.
Together with the edges $(u,w)$ and $(w',v)$, we
get a $T$-monotone path between $u$ and $v$ that consists of at most $k$ edges. 
\\\emph{Case 5: $u,v \in T_i$, for some subtree $T_i$ in $T \setminus CV_\ell$.}
The first assertion of Lemma \ref{cen} and the second assertion of Lemma \ref{base'real} imply that 
the required size $|R_i| = |R(T_i)|$ of $T_i$ is at most $\ell = \alpha'_{k-2}(n) < n$. By Lemma \ref{description}, the tree $T'_i$ that is computed at the fifth step of the algorithm
satisfies $R(T'_i) = R(T_i)$.
Hence, by the induction hypothesis for the pair $(k,|R_i|)$, the $T'_i$-monotone diameter of the 
graph $G_{T'_i} = (V(T'_i),E_i)$ that is computed at the fifth step of the algorithm is at most
$k$. It follows that
$u$ and $v$ are connected in $G_T$ by a $T'_i$-monotone path that consists of at most $k$ edges.
However, since $T'_i$ is $T_i$-monotone preserving, this path is also $T_i$-monotone, and thus also $T$-monotone.
\QED
\end{proof}

Finally, we bound the running time of the algorithm $Tree1Spanner((T,rt),n,k)$.
\begin{lemma} \label{thm3}
Let $k \ge 2$ and $n \ge 0$ be two arbitrary integers, and denote by $C_k(n)$ the worst-case running time
of Algorithm 
$Tree1Spanner((T,rt),n,k)$, where $T$ ranges over all pruned trees $T$ with required-size $n$.
Then $C_k(n) = O(n \alpha_k(n))$.
\end{lemma}
\begin{proof}
Clearly, if $n \le k+1$, then $C_k(n) = O(1)$. We may henceforth assume that $n \ge k+2$. 
\\We remark that one can compute the values of the function $\alpha'_k = \alpha'_{k}(n)$ in $O(n)$ time, for all $k \ge 2$ and $n \ge k+2$.
These values can be computed similarly to the way the values of the function $\alpha_k = \alpha_k(n)$
were computed in \cite{LaPo90}. (See also Exercise 12.7 in \cite{NS07}; further details on this technical argument are omitted.) 
In particular, computing the value of $\alpha'_{k-2}(n)$ with which $\ell$ is assigned at the first step of the algorithm
can be carried out in $O(n)$ time. 
Also, an additional time of $O(n)$ suffices to compute the set $CV_\ell$ of cut vertices at the first step of the algorithm.
The computation of the edge set $E'$ at the second step of the algorithm starts by computing a copy $\tau$ of $T$,
which can be carried out in $O(n)$ time. Another $O(n)$ time is required to go over all the vertices of $\tau$
and color the vertices of 
$CV_\ell$ in black, and the remaining vertices in white.
Computing the pruning $\tau'$ of $\tau$ also requires $O(n)$ time.
Finally, the recursive call $Tree1Spanner((\tau',rt(\tau')),|CV_\ell|,k-2)$ requires at most $C_{k-2}(|CV_\ell|)$ time.
Overall, the time needed to compute the edge set $E'$ is bounded above by $O(n) + C_{k-2}(|CV_\ell|)$.
An additional amount of $O(n)$ time is needed to compute the subtrees $T_1,\ldots,T_g$ and the corresponding pruned subtrees $T'_1,\ldots,T'_g$ at the
third and fifth steps of the algorithm, respectively. 
By Lemma \ref{bridges},  computing 
the edge set $E''$ at the fourth step of the algorithm takes another $O(n)$ time.
Finally, the time needed to compute the 
edge sets $E_1,E_2,\ldots,E_g$ at the fifth step of the algorithm is
at most $\sum_{i=1}^g C_k(|R_i|)$. We obtain the recurrence
$C_k(n) = O(n) + C_{k-2}(|CV_\ell|) + \sum_{i=1}^g C_k(|R_i|),$
where $|CV_\ell| \le \left\lfloor \frac{n}{\ell+1} \right\rfloor$, $|R_i| \le \ell = \alpha'_{k-2}(n)$,
for each index $i \in [g]$, and $\sum_{i=1}^g |R_i| \le n$. 
Hence, as in the proof of 
Lemma \ref{edgebound}, it can be
shown that $C_k(n) = O(n \alpha'_k(n)) = O(n \alpha_k(n))$.
\QED
\end{proof}

Lemmas \ref{kequal2}, \ref{kequal3}, \ref{edgebound}, \ref{thm2} and \ref{thm3} imply Theorem \ref{kequalk}.


\section{Euclidean Sparse Spanners with Bounded Diameter} \label{s:upper} 
In this section we plug the 1-spanners for tree metrics from Section \ref{1span}
on top of the dumbbell trees of \cite{ADMSS95,NS07} to obtain our construction
of Euclidean spanners.  
\begin{theorem} (``Dumbbell Theorem'', Theorem 2 in \cite{ADMSS95}, Theorem 11.9.1 in \cite{NS07}) 
Given a set $S$ of $n$ points in
$\mathbb R^d$ and a parameter $\eps > 0$, a forest
$\mathcal F$ consisting of $O(1)$ rooted trees of size $O(n)$ each can be built in $O(n \log n)$ time, having
the following properties:
1) For each tree in $\mathcal F$, there is a 1-1 correspondence
between the leaves of this tree and the points of $S$. 2) Each
internal vertex in the tree has a unique representative point, which
can be selected arbitrarily from the points in any of its descendant
leaves. 3) For any two points $u,v \in S$, there is a tree in
$\mathcal F$, so that the path formed by walking from representative
to representative along the unique path in that tree between $u$ and $v$, is a $(1+\eps)$-spanner path.
\end{theorem}


Let $S$ be a set of $n$ points in $\mathbb R^d$, let $\mathcal F$ be the forest of dumbbell trees given by the Dumbbell Theorem,
and let $k \ge 2$ be an arbitrary integer.
For each
dumbbell tree $T \in \mathcal F$, 
let $G_T$ be the 1-spanner for $T$ from Theorem \ref{kequalk} with diameter at most $k$ and $O(n \alpha_k(n))$ edges.
Our construction of Euclidean spanners is defined to be the geometric graph $\mathcal G_k(n)$  
implied by the collection of all the graphs 
$G_T$, $T \in \mathcal F$.

Since each graph $G_T$ has only $O(n \alpha_k(n))$ edges, 
the collection of $O(1)$ such graphs will also have at most $O(n \alpha_k(n))$ edges.

By the Dumbbell Theorem, the forest $\mathcal F$ of dumbbell trees can be built in $O(n \log n)$ time.
(In particular, the dumbbell trees of \cite{NS07} can be built within time $O(n \log n)$ in the
algebraic computation-tree model.)
By Theorem \ref{kequalk}, we can compute each of the graphs $G_T$ within time $O(n \alpha_k(n)) = O(n \log n)$.
Since there is a constant number of such graphs,  we get that the overall time needed to
compute our construction $\mathcal G_k(n)$ of Euclidean spanners is $O(n \log n)$.

Finally, we show that $\mathcal G_k(n)$ is a $(1+\eps)$-spanner for $S$ with diameter at most $k$. 
Consider an arbitrary pair of points $u,v \in S$.
By the Dumbbell Theorem, there is a dumbbell tree $T \in \mathcal F$, so that the geometric
path $\mathcal P_T(u,v)$ implied by the unique
path $P_T(u,v)$ between $u$ and $v$ in $T$ is a
$(1+\eps)$-spanner path. Theorem \ref{kequalk} implies that
there is a 1-spanner path $P$ for $T$ between $u$ and $v$ in $G_T$ 
that consists of at most $k$ edges. By the
triangle inequality, the weight of the corresponding geometric path
$\mathcal P$ in
 $\mathcal G_k(n)$ is no greater than the weight of $\mathcal P_T(u,v)$.
Hence, $\mathcal P$ is  a $(1+\eps)$-spanner path for $u$ and $v$ that consists of at most $k$ edges.

\begin{corollary}
For any set of $n$ points in $\mathbb R^d$, any integer $k \ge 2$
and a number $\eps
> 0$, we can compute in $O(n \log n)$ time
a $(1+\eps)$-spanner with diameter at most $k$ and $O(n \alpha_k(n))$ edges.
\end{corollary}

\section{Lower Bounds for Euclidean Steiner Spanners} \label{s:lower}
In this section we extend the lower bound of \cite{CG06} to Euclidean Steiner spanners.
\begin{theorem} \label{mainii}
Let $X$ be a set of $n$ points on the $x$-axis,
and let $H = (V,E)$, $X \subseteq V$, be a Euclidean Steiner $t$-spanner for $X$, with $t \ge 1$, having
diameter $\Lambda$ and $m$ edges.
Then $H$ can be transformed 
into a Euclidean $t$-spanner $H' = (X,E')$
with diameter at most $\Lambda$ and at most $4m$ edges.
\end{theorem}
\begin{proof}
For every point $p \in \mathbb R^d$, denote by $p(x)$ its projection onto the $x$-axis. Let $S = V\setminus X$
be the set of Steiner points of $H$, and let
$\tilde S$ be the set of all projections of the points in $S$ onto the $x$-axis, i.e., 
$\tilde S = \{v(x) ~\vert~ v \in S\}$. Also, define $\tilde V = X \cup \tilde S$, and let $\tilde H = (\tilde V,\tilde E)$
be the graph obtained from $H$ by replacing each edge $e = (u,v)$ with its projection $\tilde e = (u(x),v(x))$ onto
the $x$-axis.
Clearly, $\tilde H$ is a spanning subgraph over a superset $\tilde V$ of $X$ of points that lie on the $x$-axis, having 
$|\tilde E| = |E| = m$ edges. Also, it is easy to see that for every pair $u,v$ of points in $V$, and
every path $P = (u = v_0,v_1,\ldots,v=v_m)$ in $H$ between $u$ and $v$, the weight $w(\tilde P)$ of the corresponding path 
$\tilde P = (u(x) = v_0(x),v_1(x),\ldots,v(x) = v_m(x))$ in $\tilde H$ is no greater than the weight $w(P)$ of $P$. Hence, $\tilde H$ is a 
Euclidean Steiner $t$-spanner for $X$ over a superset $\tilde V$ of $X$ of points on the $x$-axis, having diameter at most $\Lambda$ and $m$ edges.

For every point $v \in \tilde V$, denote by $v_L$ (respectively, $v_R$) 
the point closest to $v$ among all points in $X$ that are located left (resp., right) to $v$ on the $x$-axis,
including $v$ itself.
If $v \in X$, then $v_L = v_R = v$.
If there is no point in $X$ to the left (respectively, right) of $v$, then we write $v_L = NULL$ (resp., $v_R = NULL$).
Let $\hat H$ be the graph obtained from $\tilde H$ by replacing each edge $(u,v) \in \tilde E$
with the four edges $(u_L,v_L)$, $(u_L,v_R)$, $(u_R,v_L)$,
and $(u_R,v_R)$. Notice that the resulting graph $\hat H$ may contain multiple copies of the same edge
as well as self loops, and so $\hat H$ is, in fact, a multigraph.
In addition, $\hat H$ may contain edges with one or two NULL endpoints. 
Next, we transform $\hat H$ into a simple graph $H'$ by removing from it all the multiple edges, self loops,
and edges with either one or two NULL endpoints.   
It is easy to see that no edge in the resulting graph $H'$ is incident on a Steiner point. Moreover, $H'$ contains at most $4m$
edges.  
To complete the proof of Theorem \ref{mainii}, we employ the following lemma.
\begin{lemma} \label{stl}
Let $u,v$ be an arbitrary pair of distinct points in $\tilde V$, let $u'$ be either $u_L$ or $u_R$,
and let $v'$ be either $v_L$ or $v_R$, with $u',v' \ne NULL$. Then for any path $\tilde P$ between $u$ and $v$ in $\tilde H$,
there exists a path $P'$ between $u'$ and $v'$ in $H'$, such that $|P'| \le |\tilde P|$ and $w(P') \le w(\tilde P) + \|u',u\|
+ \|v',v\|$.
\end{lemma}   
\begin{proof}
The proof is by induction on the number of edges $q = |\tilde P|$ in the path $\tilde P$.
\\\emph{Basis: $|\tilde P| = 1$.} In this case $\tilde P = (u,v)$. The proof is immediate if $u' = v'$,
and so we may assume that $u' \ne v'$. By construction, $H'$ contains
the edge $(u',v')$. Set $P' = (u',v')$. Clearly, $|P'| = |\tilde P| = 1$. Also, by the triangle inequality, 
$$w(P') ~=~  \|u',v'\| ~\le~ \|u,v\| + \|u',u\|
+ \|v',v\| ~=~ w(\tilde P) + \|u',u\|
+ \|v',v\|.$$ 
\\\emph{Induction Step:} We assume the correctness of the statement for all smaller values of $q$,
and prove it for $q$. Consider the second point $w$ on the path $\tilde P = (u,w,\ldots,v)$ between $u$ and $v$ in $\tilde H$.
Observe that either $w_L$ or $w_R$ is located on the line segment between $u'$ and $w$. (In the case $u'= w$,
we have $w_L = w_R = u' = w$.)
Denote this vertex by $w'$, and note that it is possible to have $u' = w'$, e.g., if $u' = w$.
Since $(u,w)$ is an edge in $\tilde P \in \tilde H$, it holds by construction that $(u',w')$ is an edge in $H'$. Consider the
sub-path $\tilde P_{w,v}$ of $\tilde P$ between $w$ and $v$
obtained 
by removing the first edge $(u,w)$ from $\tilde P$. It consists of $q-1$ edges, and so $|\tilde P_{w,v}| = q-1$.
Hence, by the induction hypothesis, there exists a path $P'_{w',v'}$ between $w'$ and $v'$ in $H'$,
such that $|P'_{w',v'}| \le |\tilde P_{w,v}| = q-1$ and $w(P'_{w',v'}) \le w(\tilde P_{w,v}) + \|w',w\| + \|v',v\|$. 
Since $w'$ is located on the line segment between $u'$ and $w$, it follows that $\|u',w'\| + \|w',w\| = \|u',w\|$. 
By the triangle inequality,
$\|u',w\| \le \|u',u\| + \|u,w\|$. 
Let $P'$ be the path obtained by concatenating the edge $(u',w')$ with the path $P'_{w',v'}$, i.e., $P' = (u',w') \circ P'_{w',v'}$.
Notice that $P'$ is a path in $H'$ between $u'$ and $v'$, and $|P'| = 1 + |P'_{w',v'}| \le q$.
Also, we have 
\begin{eqnarray*}w(P') &=& \|u',w'\| + w(P'_{w',v'}) ~\le~ \|u',w'\| + w(\tilde P_{w,v}) + \|w',w\| + \|v',v\|
\\ &=& \|u',w\| + w(\tilde P_{w,v})  + \|v',v\|  ~\le~ \|u',u\| + \|u,w\| + w(\tilde P_{w,v}) + \|v',v\|
\\ &=& w(\tilde P) +  \|u',u\| + \|v',v\|. \inQED
\end{eqnarray*}
\end{proof}
Lemma \ref{stl} implies that for any two points in $X$, there is a $t$-spanner path in $H'$ that consists of at most
$\Lambda$ edges. Thus $H'$ is a Euclidean $t$-spanner for $X$ with diameter
at most $\Lambda$ and at most $4m$ edges.
\QED
\end{proof}
Chan and Gupta \cite{CG06} proved that for any $\eps > 0$, there exists a set $S_\eps$ of $n$ points on the $x$-axis,
where $n$ is an arbitrary power of two, 
for which any Euclidean $(1+\eps)$-spanner with at most
$m$ edges has diameter at least $\Omega(\alpha(m,n))$. Theorem \ref{mainii} enables us to extend
the lower bound of \cite{CG06} to Euclidean Steiner spanners. 
\begin{corollary}
For any $\eps > 0$, there exists a
set of $n$ points on the $x$-axis, for which any Euclidean (possibly Steiner) $(1+\eps)$-spanner with at most $m$ edges has 
diameter at least $\Omega(\alpha(m,n))$.
\end{corollary}
\begin{proof}
The statement is trivial if $\alpha(m,n) = O(1)$. We henceforth assume that $\alpha(m,n)$ is super-constant.

Let $S_\eps$ be the aforementioned set of $n$ points for which the lower bound of \cite{CG06} holds,
and suppose for contradiction that there exists a Euclidean Steiner $(1+\eps)$-spanner $H$ for $S_\eps$
with at most $m$ edges and diameter $\Lambda = o(\alpha(m,n))$. By Theorem \ref{mainii}, we can transform
$H$ into a Euclidean $(1+\eps)$-spanner $H'$ for $S_\eps$, having diameter $\Lambda' \le \Lambda = o(\alpha(m,n))$
and at most $4m$ edges. 
However, the lower bound of \cite{CG06}
implies that the diameter $\Lambda'$ of $H'$ is at least $\Omega(\alpha(4m,n))$.
Using the observation that $\alpha(4m,n) \ge \alpha(m,n)-4$, for all $m \ge n$, we conclude that
$\Lambda' = \Omega(\alpha(m,n)-4) = \Omega(\alpha(m,n))$, 
yielding a contradiction. 
\QED
\end{proof}

\section{Acknowledgments}
The author is grateful to Michael Elkin and Michiel Smid for helpful discussions.

\bibliographystyle{latex8}
\bibliography{latex8}

\clearpage
\pagenumbering{roman}
\appendix
\centerline{\LARGE\bf Appendix}

\section{Proof of Lemma \ref{base'real}} \label{acker}
This section is devoted to the proof of Lemma \ref{base'real}.

We start with proving the following claim.

\begin{claim} \label{claimii}
(1) The function $\alpha'_2 = \alpha'_2(n)$ is monotone non-decreasing with $n$.
~(2) For all $n \ge 1$, $\alpha'_2(n) \le n-1$. Moreover, if $n \ge 6$, then $\alpha'_2(n) \le n-2$.
~(3) For all $n > 10$, $\alpha'_2(n) \le \alpha'_0(n)$.
~(4) For all $n \ge 0$, $\alpha'_4(n) \le \alpha'_2(n)$.
\end{claim}
\begin{proof}
We first prove that $\alpha'_2 = \alpha'_2(n)$ is monotone non-decreasing with $n$. Specifically, we show 
that for all $n \ge 0$: $\alpha'_2(m) \le \alpha'_2(n)$, for any $m \le n$.
The proof is by induction on $n$. The basis $n \le 3$ can be easily verified.
\\\emph{Induction Step:} We assume the correctness of the statement for all
smaller values of $n$, $n \ge 4$, and prove it for $n$.
By definition, $\alpha'_2(n) = 2 + \alpha'_2(\alpha'_0(n)) = 2 + \alpha'_2(\lceil n/2 \rceil)$.
It is easy to see that for $m \le 3$, $\alpha'_2(m) \le 2$, and so $\alpha'_2(m) \le 2 \le 2 + \alpha'_2(\lceil n/2 \rceil) = \alpha'_2(n)$. 
\\We  henceforth
assume that $4 \le m \le n$. Thus, by definition, $\alpha'_2(m) = 2 + \alpha'_2(\alpha'_0(m)) =   2 + \alpha'_2(\lceil m/2 \rceil)$.
Since $4 \le m \le n$, we have $\lceil m/2 \rceil \le \lceil n/2 \rceil < n$. By the induction hypothesis for $\lceil n/2 \rceil$,
 $\alpha'_2(\lceil m/2 \rceil) \le \alpha'_2(\lceil n/2 \rceil)$. It follows that
 $$\alpha'_2(m) ~=~ 2 + \alpha'_2(\lceil m/2 \rceil) ~\le~ 2 + \alpha'_2(\lceil n/2 \rceil) ~=~ \alpha'_2(n).$$

We proceed by proving
the second assertion.
It is easy to verify that $\alpha'_2(n) = n-1$, for all $1 \le n \le 5$. 
\\Next, we prove by induction on $n$ that for all $n \ge 6$, it holds that $\alpha'_2(n) \le n-2$.
The basis $n = 6$ can be easily verified.
\\\emph{Induction Step:} We assume the correctness of the statement for all smaller
values of $n$, $n \ge 7$, and prove it for $n$.
By definition, $\alpha'_2(n) = 2 + \alpha'_2(\alpha'_0(n)) = 2 + \alpha'_2(\lceil n/2 \rceil)$.
Since $n \ge 7$, $4 \le \lceil n/2 \rceil < n$. If $\lceil n/2 \rceil \le 5$, then 
$\alpha'_2(\lceil n/2 \rceil) \le \lceil n/2 \rceil -1$. Otherwise, $\lceil n/2 \rceil \ge 6$, and
by the induction hypothesis for $\lceil n/2 \rceil$, we have $\alpha'_2(\lceil n/2 \rceil) \le \lceil n/2 \rceil -2$.
In any case, it holds that $\alpha'_2(\lceil n/2 \rceil) \le \lceil n/2 \rceil -1$.
Consequently, $$\alpha'_2(n) ~=~ 2 + \alpha'_2(\lceil n/2 \rceil) ~\le~ 1+ \lceil n/2 \rceil ~\le~ n-2.$$
(The last inequality holds for all $n \ge 6$.) 

To prove the third assertion,
consider an arbitrary integer $n \ge 11$.
By definition, $\alpha'_2(n) = 2 + \alpha'_2(\alpha'_0(n)) = 2 + \alpha'_2(\lceil n/2 \rceil)$.
Since $\lceil n/2 \rceil \ge 6$, the second assertion of this claim yields $\alpha'_2(\lceil n/2 \rceil)
\le \lceil n/2 \rceil -2$, and so $$\alpha'_2(n) ~=~  2 + \alpha'_2(\lceil n/2 \rceil) ~\le~ \lceil n/2 \rceil ~=~ \alpha'_0(n).$$

The proof of the fourth assertion is by induction on $n$. For all $0 \le n \le 10$ the statement can be verified by brute force.
\\\emph{Induction Step:} We assume the correctness of the statement for all smaller
values of $n$, $n \ge 11$, and prove it for $n$.
By definition, for all $n \ge 11$, 
$\alpha'_4(n) = 2 + \alpha'_4(\alpha'_2(n))$
and $\alpha'_2(n) = 2 + \alpha'_2(\alpha'_0(n))$. By the third assertion of this claim, we know that
$\alpha'_2(n) \le \alpha'_0(n)$. Hence, by the first assertion of this claim, $\alpha'_2(\alpha'_0(n)) \ge \alpha'_2(\alpha'_2(n))$.
The second assertion of this claim implies that $\alpha'_2(n) < n$, and so by the induction hypothesis for $\alpha'_2(n)$, 
we have $\alpha'_4(\alpha'_2(n)) \le  \alpha'_2(\alpha'_2(n))$.
Altogether, $$\alpha'_4(n) ~=~ 2 + \alpha'_4(\alpha'_2(n))
~\le~ 2 + \alpha'_2(\alpha'_2(n)) ~\le~ 2 + \alpha'_2(\alpha'_0(n)) ~=~ \alpha'_2(n).$$ 
\QED
\end{proof}

We now turn to the proof of Lemma \ref{base'real}.
\ignore{
The next lemma is analogous to Lemma \ref{base}.
\begin{lemma} \label{base'}
\begin{enumerate}
\item For all $k \ge 2$, the function $\alpha'_k$ is monotone, i.e., for all $n \ge 0$:
$\alpha'_{k}(m) \le \alpha_{k}(n)$, for any $m \le n$.
\item For all $k \ge 2$ and $n \ge 0$, $\alpha'_{k+2}(n) \le \alpha'_{k}(n)$.
\item For all $k \ge 2$ and $n \ge 1$, $\alpha'_{k}(n) < n$.
\end{enumerate}
\end{lemma}} 
\\The three assertions of the lemma are proved by double induction on $k$ and $n$.
We restrict the attention to even values of $k$.
The argument for odd values of $k$ is similar, and is thus omitted.
\\For technical convenience,  we will prove the first assertion of the lemma in the sequel by showing that
$\alpha'_{k}(m) \le \alpha'_{k}(n)$, for an arbitrary integer $m \le n$.  
\\The case $k=2$ follows from Claim \ref{claimii}. 
\\Suppose next that $n \le k+1$. By definition, $\alpha'_{k}(n) = \alpha_{k}(n)$, for any $n \le k+1$, 
and so the three assertions of the lemma follow from Lemma \ref{basereal}.
\\\emph{Induction Step:} We assume that for an arbitrary pair $(k,n)$, $n \ge k+2$, $k \ge 4$, 
the three assertions of the lemma hold for all pairs $(k',n')$,
with either $k' < k$ or both $k' = k$ and $n' < n$, and prove it for the pair $(k,n)$.

Consider an arbitrary integer $m \le n$.
We first show that $\alpha'_{k}(m) \le \alpha'_{k}(n)$, thus proving the first assertion of the lemma. 
If $m \le k+1$, then $\alpha'_{k}(m) = \alpha_{k}(m)$.
Clearly, $\alpha_{k}(n) \le \alpha'_{k}(n)$. By the first assertion of Lemma \ref{basereal},
$\alpha_{k}(m) \le \alpha_{k}(n)$, and so $$\alpha'_{k}(m) ~=~ \alpha_{k}(m) ~\le~ \alpha_{k}(n) ~\le~ \alpha'_{k}(n).$$
Otherwise, we have $k+2 \le m \le n$. Hence, by definition, $\alpha'_{k}(n) = 2 + \alpha'_{k}(\alpha'_{k-2}(n))$
and $\alpha'_{k}(m) = 2 + \alpha'_{k}(\alpha'_{k-2}(m))$. By the first and second assertions of the induction hypothesis for the pair $(k-2,n)$, we have
$\alpha'_{k-2}(m) \le \alpha'_{k-2}(n)$
and $\alpha'_{k-2}(n) < n$, respectively. Hence, the first assertion of the induction hypothesis for the pair $(k,\alpha'_{k-2}(n))$
implies that
$\alpha'_{k}(\alpha'_{k-2}(m)) \le \alpha'_{k}(\alpha'_{k-2}(n))$.
 Altogether, $$\alpha'_{k}(m) ~=~ 2 + \alpha'_{k}(\alpha'_{k-2}(m))
 ~\le~ 2 + \alpha'_{k}(\alpha'_{k-2}(n)) ~=~ \alpha'_{k}(n).$$ 

The second and third assertions of the induction hypothesis for the pair $(k-2,n)$ imply that  
\begin{equation} \label{dah} \alpha'_{k}(n) ~\le~ \alpha'_{k-2}(n) ~<~ n,
\end{equation} thus proving the second assertion of the lemma.

We next prove the third assertion of the lemma.
Suppose first that $k+2 \le n \le k+3$. In this case $\alpha'_{k+2}(n) = \alpha_{k+2}(n)$.
Also, we have $\alpha'_{k}(n) \ge \alpha_{k}(n)$. By the third assertion of Lemma \ref{basereal},
$\alpha_{k+2}(n) \le \alpha_{k}(n)$, yielding $$\alpha'_{k+2}(n) ~=~ \alpha_{k+2}(n)
~\le~ \alpha_{k}(n) ~\le~ \alpha'_{k}(n).$$
Otherwise, $n \ge k+4$. In this case, $\alpha'_{k+2}(n) = 2+\alpha'_{k+2}(\alpha'_{k}(n))$
and $\alpha'_{k}(n) = 2+\alpha'_{k}(\alpha'_{k-2}(n))$.
Equation (\ref{dah}) and the first assertion of the induction hypothesis for the pair $(k,\alpha'_{k-2}(n))$ imply that 
$\alpha'_{k}(\alpha'_{k}(n)) \le \alpha'_{k}(\alpha'_{k-2}(n))$.
Also, equation (\ref{dah}) and the third assertion of the induction hypothesis for the pair $(k,\alpha'_{k}(n))$ imply
that $\alpha'_{k+2}(\alpha'_{k}(n)) \le \alpha'_{k}(\alpha'_{k}(n))$,
yielding $\alpha'_{k+2}(\alpha'_{k}(n)) \le \alpha'_{k}(\alpha'_{k}(n)) \le \alpha'_{k}(\alpha'_{k-2}(n))$.
Altogether, $$\alpha'_{k+2}(n) ~=~ 2+\alpha'_{k+2}(\alpha'_{k}(n)) ~\le~ 2 + \alpha'_{k}(\alpha'_{k-2}(n)) ~=~ \alpha'_{k}(n).$$


\section{Proof of Lemma \ref{finallyreal}} \label{acker2}
This section is devoted to the proof of Lemma \ref{finallyreal}
\ignore{
The following lemma can be easily verified.
\begin{lemma} \label{base}
\begin{enumerate}
\item For all $k \ge 0$,
the function $\alpha_k = \alpha_k(n)$ is monotone non-decreasing with $n$,
i.e., for all $n \ge 0$:
$\alpha_k(m) \le \alpha_k(n)$, for any $m \le n$.
\item For all $k \ge 0$ and $n \ge 0$, $\alpha_{k+2}(n) \le \alpha_{k}(n)$.
\item \begin{itemize} \item For all $n \ge 2$, $\alpha_{0}(n) < n$.
\item For all $n \ge 3$, $\alpha_{1}(n) < n$.
\item For all $k \ge 2$ and $n \ge 1$, $\alpha_{k}(n) < n$.
\end{itemize}
\end{enumerate}
\end{lemma} 
}

We start with proving the following claim.
\begin{claim} \label{auxim}
For all $k \ge 0$ and $n \ge 0$, 
$\alpha_{k}(2(n+2)) <  2 (\alpha_k(n) + 2)$. 
\end{claim}
\begin{proof}
The proof is by double induction on $k$ and $n$. 
We restrict the attention to even values of $k$. 
 The argument for odd values of $k$ is similar,
and is thus omitted. 
\\Consider first the case $k = 0$.
By definition, $\alpha_0(n) = \lceil n/2 \rceil$, for all $n \ge 0$.
Hence, $\alpha_0(2(n+2)) =  \lceil (2(n +2))/2 \rceil <  2(\lceil n/2 \rceil + 2) = 2 (\alpha_0(n) + 2)$.
\\Suppose next that $n < 2$. Notice that $2(n+2) \le 6$ and $\alpha_0(6) = \lceil 6/2 \rceil = 3$. By the third assertion of Lemma \ref{basereal},  $\alpha_{k}(6) \le \alpha_0(6) = 3$. Hence, by the first assertion of Lemma \ref{basereal},
$\alpha_{k}(2(n +2)) \le \alpha_{k}(6) \le 3 < 2 (\alpha_{k}(n) + 2)$. 
\\\emph{Induction Step:} We assume that for an arbitrary pair $(k,n)$, $n \ge 2$, $k \ge 2$,
the statement holds for all pairs $(k',n')$,
with either $k' < k$ or both $k' = k$ and $n' < n$, and prove it for the pair $(k,n)$.
\\Since $n \ge 2$, we have $2(n+2) \ge 8$,
and so Lemma \ref{NSreal} implies that $\alpha_{k}(2(n+2)) = 1 + \alpha_{k}(\alpha_{k-2}(2(n+2)))$. 
By the induction hypothesis for the pair $(k-2,n)$, we have
$\alpha_{k-2}(2(n+2)) < 2(\alpha_{k-2}(n)+2)$. Hence, the first assertion of Lemma \ref{basereal} implies
that
$\alpha_{k}(\alpha_{k-2}(2(n+2))) \le \alpha_{k}(2(\alpha_{k-2}(n) + 2))$.
Since $n \ge 2$, the second assertion of Lemma \ref{basereal} implies that
 $\alpha_{k-2}(n) < n$. Hence, by the induction hypothesis for the pair $(k,\alpha_{k-2}(n))$,
 it follows that $\alpha_{k}(2(\alpha_{k-2}(n)+2)) < 2(\alpha_{k}(\alpha_{k-2}(n))+2)$.
 Altogether, \begin{eqnarray*}
 \alpha_{k}(2(n+2)) &=& 1 + \alpha_{k}(\alpha_{k-2}(2(n+2)))
 ~\le~ 1 + \alpha_{k}(2(\alpha_{k-2}(n) + 2)) \\ &<& 1 + 2(\alpha_{k}(\alpha_{k-2}(n)) + 2)
 ~=~ 1+2(\alpha_{k}(n)+1) ~<~ 2(\alpha_{k}(n)+2). \inQED
 \end{eqnarray*} 
\end{proof}

We are now ready to prove Lemma \ref{finallyreal}.
\\The first assertion of Lemma \ref{base'real} implies that the function $\alpha'_k = \alpha'_k(n)$ is monotone non-decreasing with $n$,
for all $k \ge 2$ and $n \ge 0$. The functions $\alpha'_0 = \lceil n/2 \rceil$ and $\alpha'_1 = \lceil \sqrt{n} \rceil$ are monotone non-decreasing with $n$ as well.
Consequently, $\alpha'_k(n) \le \alpha'_{k}(2(n+2))$, for all $k \ge 0$ and $n \ge 0$.
Hence, Lemma \ref{finallyreal} follows as a corollary of the following lemma.
\begin{lemma} \label{finally}
For all $k \ge 0$ and $n \ge 0$, 
$\alpha'_{k}(2(n+2)) \le  2 (\alpha_k(n) + 2)$. 
\end{lemma}
\begin{proof}
The proof is by double induction on $k$ and $n$.
We restrict the attention to even values of $k$.
The argument for odd values of $k$ is similar,
and is thus omitted. 
\\For the case $k=0$, we have
by definition $\alpha'_{0}(n) = \alpha_0(n)$, for all $n \ge 0$, and so the statement follows from Claim \ref{auxim}.
We henceforth assume that $k \ge 2$.
\\Next, consider the case $n < 2$.
In this case, $2(n+2) \le 6$ and $\alpha_{k}(n) = 0$.
Also, notice that $\alpha'_2(6) = 2+ \alpha'_2(3) = 2 + \alpha_2(3) = 4$.
The third assertion of Lemma \ref{base'real} implies that $\alpha'_{k}(6) \le \alpha'_{2}(6)$.
Hence, by the first assertion of Lemma \ref{base'real}, $$\alpha'_{k}(2(n+2)) ~\le~ \alpha'_{k}(6)
~\le~  \alpha'_2(6) ~=~ 4 ~=~ 2 (\alpha_{k}(n) + 2).$$
Suppose next that $2(n+2) \le k+1$.
In this case, by definition $\alpha'_{k}(2(n+2)) = \alpha_{k}(2(n+2))$,
and so the statement follows from Claim \ref{auxim}.
\\\emph{Induction Step:} We assume that for an arbitrary pair $(k,n)$, $n \ge 2$, $k \ge 2$, $2(n+2) \ge k+2$,
the statement holds for all pairs $(k',n')$,
with either $k' < k$ or both $k' = k$ and $n' < n$, and prove it for the pair $(k,n)$.
Since $2(n+2) \ge k+2$,   
we have by definition $\alpha'_{k}(2(n+2)) = 2 + \alpha'_{k}(\alpha'_{k-2}(2(n+2)))$. 
By the induction hypothesis for the pair $(k-2,n)$, we have
$\alpha'_{k-2}(2(n+2)) \le 2(\alpha_{k-2}(n)+2)$. Hence, by the first assertion of Lemma \ref{base'real},
$\alpha'_{k}(\alpha'_{k-2}(2(n+2))) \le \alpha'_{k}(2(\alpha_{k-2}(n) + 2))$.
Since $n \ge 2$, the second assertion of Lemma \ref{basereal} yields
 $\alpha_{k-2}(n) < n$. Hence, the induction hypothesis for the pair $(k,\alpha_{k-2}(n))$
 implies that $\alpha'_{k}(2(\alpha_{k-2}(n)+2)) \le 2(\alpha_{k}(\alpha_{k-2}(n))+2)$.
 Altogether, \begin{eqnarray*}
 \alpha'_{k}(2(n+2)) &=& 2 + \alpha'_{k}(\alpha'_{k-2}(2(n+2)))
 ~\le~ 2 + \alpha'_{k}(2(\alpha_{k-2}(n) + 2)) \\ &\le& 2 + 2(\alpha_{k}(\alpha_{k-2}(n)) + 2)
 ~=~ 2+2(\alpha_{k}(n)+1) ~=~ 2(\alpha_{k}(n)+2). \inQED
 \end{eqnarray*} 
\end{proof}

\end{document}